\documentclass[a4paper]{article}

\usepackage[T1]{fontenc}
\usepackage[utf8]{inputenc}
\usepackage{float}
\usepackage{graphicx}
\usepackage{subcaption} 
\usepackage[thinlines]{easytable}
\usepackage{hyperref}
\usepackage{wrapfig}
\usepackage{dingbat}
\usepackage{color}
\usepackage{refcount}
\usepackage[table]{xcolor}

\usepackage[toc,page]{appendix}

\usepackage{indentfirst}

\usepackage{amsmath,amssymb,amsthm}
\usepackage{mathtools}
\usepackage[english]{babel}

\usepackage{natbib}

\makeatletter
\@addtoreset{equation}{section}

\makeatother

\theoremstyle{Theorem}
\newtheorem{theo}{Theorem}
\theoremstyle{Lemma}
\newtheorem{lem}[theo]{Lemma}
\newtheorem{rem}[theo]{Remark}
\newtheorem{assumption}[theo]{Assumptions}

\theoremstyle{Definition}
\newtheorem{defi}[theo]{Definition}

\theoremstyle{Corollary}

\DeclareFontFamily{U}{matha}{\hyphenchar\font45}
\DeclareFontShape{U}{matha}{m}{n}{
      <5> <6> <7> <8> <9> <10> gen * matha
      <10.95> matha10 <12> <14.4> <17.28> <20.74> <24.88> matha12
      }{}
\DeclareSymbolFont{matha}{U}{matha}{m}{n}
\DeclareFontSubstitution{U}{matha}{m}{n}

\DeclareFontFamily{U}{mathx}{\hyphenchar\font45}
\DeclareFontShape{U}{mathx}{m}{n}{
      <5> <6> <7> <8> <9> <10>
      <10.95> <12> <14.4> <17.28> <20.74> <24.88>
      mathx10
      }{}
\DeclareSymbolFont{mathx}{U}{mathx}{m}{n}
\DeclareFontSubstitution{U}{mathx}{m}{n}

\DeclareMathDelimiter{\vvvert}{0}{matha}{"7E}{mathx}{"17}

\newcommand{\N}{{\mathbb N}}
\newcommand{\Z}{{\mathbb Z}}

\newcommand{\R}{{\mathbb R}}

\newcommand{\PCA}{\textnormal{PCA}}

\newcommand{\bl}{r}

\newcommand{\const}{\kappa}

\newcommand{\be}{\begin{equation}} 
\newcommand{\Ee}{\end{equation}}

\numberwithin{theo}{section}

\begin{document}

\title{Testing separability for continuous functional data}

\author{Holger Dette, Gauthier Dierickx, Tim Kutta}

\maketitle
\begin{abstract}
\noindent Analyzing the covariance structure of data is a fundamental task of statistics. While this task is simple for low-dimensional observations, it becomes challenging for more 
intricate objects, such as multivariate functions. Here, the covariance  can be so complex that just saving a non-parametric estimate is impractical and structural assumptions are necessary to tame the model. One popular assumption for space-time data is separability of the covariance into purely spatial and temporal factors. \\ 
In this paper, we present a new test for separability in the context of dependent functional time series.  While most of the related work studies functional data in a Hilbert space of square integrable functions,  we model the observations as objects in the space of 
continuous functions equipped with the supremum norm.  We argue that this (mathematically challenging) setup enhances interpretability for users and is more in line with practical  preprocessing.
 Our test 
statistic measures the maximal deviation between the estimated covariance kernel and a separable  approximation.
Critical values are obtained by  a non-standard multiplier bootstrap for dependent data. We prove the statistical validity of 
  our approach  and   demonstrate its practicability in  a simulation study and a data example.

\end{abstract}



\bigskip

\section{ Introduction } 
    \noindent
    Over the last decades, the analysis of high-dimensional space-time data has become a cornerstone of geostatistics. New technologies allow the collection of high-frequency and high-resolution measurements for   variables such as temperature, magnetic fields or pollutant concentrations (see, for example, \cite{GroKokZhuSoj2012,Aue2018,King2018}). One way to analyze such data, is to smooth it over space and time, which yields time series of spatio-temporal processes. This approach of reconstructing and analyzing random functions follows the paradigm of functional data analysis (FDA) and has recently gained attraction in the geostatistical community (for overviews see, e.g., \cite{Gen2020} and the monograph of \cite{MaGi2022}). 
    
    While the analysis of such processes promises deep scientific insights, their high complexity can push statistical methods to their limit. For example, commonly used   tools such as PCA and Kriging hinge on an approximation of the processes' covariance operator - an object that can be too massive to be stored or to be inverted for the purpose of prediction. For a concise overview of these computational challenges, we refer to  Table~1 in \cite{MasSarPan20}. To reduce complexity, many works impose structural assumptions on the covariance, one of them being \textit{separability}.
    
    Roughly speaking, separability states that the covariance of a space-time process can be decomposed into a purely temporal and a purely spatial component. The implied elimination of space-time interactions cuts the number of model parameters and makes the covariance tractable again (see \cite{Gen07}). Besides, separability entails a product structure for the principal components, facilitating the construction of estimators and inference methods (see \cite{GroKokZhuSoj2012,GroKokRei2016}). To rigorously define separability, consider a stochastic process $\{ X(s,t) ~|~s \in K_1, t \in K_2\} $, with one argument in a spatial domain $K_1$ and one in a temporal domain $K_2$. 
    Then, under suitable conditions (see, for example, \cite{JanKaiBook15}),  its covariance operator can be defined   point-wise as
    $$
    C(s,t,s',t') 
    := 
    \mathbb{ E } \left[ (X(s,t)-\mathbb{ E }  X(s,t))(X(s',t')-\mathbb{ E }  X(s',t'))\right].
    $$
    We call $C$ separable, if there exist two functions $C_1$ (spatial), $C_2$ (temporal), such that
    $$
    C(s,t,s',t') = C_1(s,s') \cdot C_2(t, t') \qquad \forall s,s' \in K_1, \,\, t,t' \in K_2.
    $$
    As pointed out before, the product structure of separability prunes model parameters, making the model  statistically and computationally more tractable.
    
    However, separability is not a free lunch for space-time data. Indeed, erroneously assuming a separable model can lead to inconsistent estimates and biased inference results. This point is crucial, as separability is rarely self-evident, as noticed by many authors  (see  \cite{ScaMar2005, Gen07,  AstPigTav17}, among many others). 
    To address this issue, statistical tests have been proposed to examine separability, such as for finite dimensional data by \cite{MatYaj2004,ScaMar2005,Fue2006,CruFer-CasGon-Man2010}. More recently, nonparametric methods tailored to functional observations  have been  devised by \cite{AstPigTav17,ConKokRei17,ConKokRei18,BagDet20,DetDieKut21}. While these latter works differ in terms of their inference strategies, they share the mathematical setup of modelling observations in a space of square integrable functions. This approach is standard in FDA (see the monographs of  \cite{HorKokBook12, hsingeubank2015}) and provides the most immediate extension of finite dimensional methodology to function spaces. Nevertheless, the choice of $L^2$-spaces is usually more informed by mathematical convenience, than by practicability. Indeed, as many undergraduate textbooks point out, the $L^2$-distance is hard to geometrically interpret and often stuns the novice by its unintuitive notion of convergence. More specific to FDA, basing a theory on $L^2$-spaces usually ignores structural features of the functions, such as continuity. Notice that the bulk of FDA relies on continuous, non-parametric curve estimation as preprocessing.  In such cases, insisting on an $L^2$-framework can  feel disjoint from a user's intuition and defy heuristic interpretations of  results. 
    
    Following recent work of    \cite{degras2011,caoyangtodem2012,degras2017}
    and
    \cite{DetKokAue20}, we opt instead to conduct FDA on the  in our opinion  more natural space of continuous functions equipped with the supremum norm. On this more intricate space, we study functional space-time processes and advance a new test for separability of the covariance, which is based on an estimate of the maximum deviation between the covariance operator  and an approximation by a covariance operator from a separable process.
    We combine  the  
    profound theory of  weak convergence of stochastic processes   (see,  for example, \cite{vdVWelBook96} or \cite{GinNikBook2016} among many others)
    with  new differentiability results for  these   separability measures to study  the asymptotic properties of  the  corresponding estimates. To improve the finite sample properties and to avoid 
    the estimation of complicated nuisance parameters 
   we propose a multiplier bootstrap for dependent data  as a more practicable alternative. 
 In particular, our results hold 
    under     weaker dependence and stationarity assumptions compared to previous works.  
    
    The rest of this paper is organized as follows. Section~\ref{Sec_MathConc} provides a brief introduction to random variables on the space of continuous functions and the notion of separability. Subsequently, we discuss three methods from the related literature to approximate a covariance kernel $C$ by a separable version (we develop theory for all three approximations at a later point). In Section~\ref{Sec_StatInfContCovKer}, we present the test statistic for the hypothesis of separability and prove its  weak convergence.
    To approximate the limiting distribution, we discuss a non-standard multiplier bootstrap for dependent data. Section~\ref{Sec_finite_sample} is dedicated to the finite sample properties of our test, which we study by virtue of simulations as well as a data example. Finally, all proofs and technical details are deferred to the Appendix.

\section{Mathematical concepts} \label{Sec_MathConc}

\noindent  
	In this section, we lay the mathematical foundations of FDA in the space of continuous functions. 
	Section~\ref{Subsec_SpaceContFunc} begins with a review of random, continuous functions, as well as basic concepts, such as functional expectations  and covariance kernels. Subsequently, we define the model assumption of separability, which is the focus of our below statistical analysis. In Section~\ref{Subsec_SepaApprox}, we discuss three methods to approximate a kernel $A$ by a separable version $A^{\mathbf{x}}$ and show that each approximation map ($A \mapsto A^{\mathbf{x}}$) is well-defined and differentiable (see Theorem~\ref{Theo_Frechet}).

	\subsection{Mathematical preliminaries }\label{Subsec_SpaceContFunc}
	
	\noindent 
	Let
	$$
	\mathcal{ C }(K) := \big\{f: K \to \mathbb{R}\, | \, f \,\, \textnormal{continuous} \big\},
	$$  	
	denote  the space 
	of continuous, real valued functions defined  on a non-empty, compact set $K \subset \mathbb{R}^d$,
	which equipped with the common supremum norm (or "sup-norm" for short)
	\begin{equation} \label{def_supnorm}
	\| f\| := \sup_{t \in K}|f(t)|
	\end{equation}
	is a Banach space. If there is no danger of confusion, we sometimes refer to a function $f \in \mathcal{ C }(K)$ by its evaluation $f(t)$. 
	Moreover, 
	we  also use the notation $\|\cdot\|$ to refer to the max-norm (maximum absolute entry) for matrices or vectors
	(corresponding to a finite set $K$ in \eqref{def_supnorm}).

	Letting $(\Omega, \mathcal{A}, \mathbb{P})$ denote a complete probability space, we call a map $X:(\Omega, \mathcal{A}, \mathbb{P}) \to \mathcal{ C }(K)$ a random $\mathcal{ C }(K)$-valued function, if $X$ is Borel-measurable w.r.t.\ the sup-norm. We point out that Borel-measurability of $X$ (as a Banach space valued function) is equivalent to measurability of the real valued marginals $X(t):(\Omega, \mathcal{A}, \mathbb{P}) \to \R$ for all $t \in K$. Supposing that the first absolute moment of $X$ exists, in the sense that $\mathbb{ E } \|X\|<\infty$, the expectation of $X$  is well-defined (in the Bochner sense), with $\mathbb{ E }  X \in \mathcal{ C }(K)$  and equal to the point-wise expectation $\mathbb{ E }[ X(t)]$ for any $t \in K$. If the stronger moment condition $\mathbb{ E } \|X\|^2<\infty$ holds, we can define the expectation of the product function $\{(X(s)-\mathbb{ E } X(s))(X(t)-\mathbb{ E } X(t))\}_{s,t \in K}$ on the space $\mathcal{ C }(K^2):=\mathcal{ C }(K \times K)$. We call this the \textit{covariance kernel} and define it point-wise as
 $$
	C(s,t) := \mathbb{ E } [X(s)-\mathbb{ E } X(s)][X(t)-\mathbb{ E } X(t)].
	$$
	For further details on defining moment of Banach space valued random variables we refer to \cite{JanKaiBook15}.
	 As common in the study of functional data, we sometimes consider for a continuous kernel function 
	$A \in \mathcal{ C }(K^2)$ 
	the corresponding integral operator $\{f(t)\}_{t \in K} \mapsto \{ \int_K A( s, t ) f(s) ds\}_{t \in K}$. 
    For ease of notation, we usually identify the integral operator with the associated kernel (notice that the relation is one-one) and define the expression 
    $A[f]
    :=
    \{ \int_K A( s, t ) f( s ) ds\}_{t \in K}$. In particular, we call the integral operator associated with $C$ the \textit{covariance operator}.
    Finally, we observe that $\mathcal{ C }(K)$ can be understood as a subspace of $L^2(K)$, the space of square integrable functions, equipped with the canonical semi-norm 
    $\|f\|_{ L^2 } := \{\int_K f(s)^2 ds\}^{1/2}$. 
    Recall that the $L^2$-topology is weaker than the uniform topology, since we consider functions with compact support.
    
    In the next step, we consider the mathematical property of \textit{separability}. Let $K_1 \subset \R^d, K_2 \subset \R^q$ be compact and non-empty sets. Then we say that a kernel 
    $A \in \mathcal{ C }((K_1 \times K_2)^2)$ 
    is separable, if there exist kernels 
    $A_i \in \mathcal{ C }(K_i^2)$ for $i=1,2$, s.t. 
\begin{equation*} %\label{Eq_def_sep}
    	A(s,t,s',t') 
    = 
        A_1(s,s') \, A_2 (t,t') 
    \qquad 
        \forall s,s' \in K_1, \,\, t,t' \in K_2. 
\end{equation*}
    Later, we will study the case of a separable covariance kernel $C$ of a random function $X$. In order to assess separability, we will employ an estimator $\hat C_N$ and compare it to a separable approximation. The construction of such separable approximations is the subject of the next section.

\subsection{Separable approximations}\label{Subsec_SepaApprox}

    \noindent In this section, we discuss the approximation of a kernel function 
    $A \in \mathcal{ C }((K_1 \times K_2)^2)$ 
    by the (separable) product of two functions $B = A_1 \cdot A_2$ with $A_i \in \mathcal{ C }(K_i^2)$ and $i=1,2$. 
    This approximation is key to our subsequent analysis, as it characterizes separability by the vanishing goodness-of-fit measure 
    $\|A-B\|=0$. How then should we construct a separable approximation?
    The most natural way might be to optimize over all separable maps w.r.t.\ the norm $\|\cdot \|$, i.e.,  to find 
$$
    B^{\textnormal{opt}} \in  \textnormal{argmin} \{\|A-B\|: B \, \textnormal{separable}\}.
$$
    Unfortunately, this type of approximation is computationally infeasible. To illustrate this point, let us consider the analogue problem of finding an optimal separable approximation for a $d \times d$-matrix w.r.t.\ to the max-norm (also denoted by $\|\cdot \|$). This problem reflects the search for an optimal, separable approximation for discretized a version of $A$ (as it would be saved on a computer in a real application). In the following Lemma "$\otimes$" denotes the well-known Kronecker product for matrices.

\begin{rem} \label{Lem_NP-Comp-of-MiniProb}
    Let $M \in \mathbb{R}^{d \times d}$ be a real valued matrix with $d=pq$ and $p, q \in \N$. Then the problem of finding a solution to the following minimization problem 
$$
    ``\textnormal{minimize} \,\, \|M- M_1 \otimes M_2\|  \quad \textnormal{with} \,\, 
    \textnormal{ $M_1 \in \mathbb{R}^{p \times p}$, \, $M_2 \in \mathbb{R}^{q \times q}$''},
$$
    is $NP$-complete. 
    \end{rem}
    
    \begin{proof}
        The problem of finding optimal rank-$1$-approximations with regard to the max-norm is an NP-complete problem, according to \cite{GilShi17}. According to Section~2 of \cite{Gen07} the problem of finding separable approximations is equivalent to finding rank-$1$-approximations.
    \end{proof}
    
    \noindent The notion of \textit{$NP$-completeness} in the above Remark~\ref{Lem_NP-Comp-of-MiniProb} refers to a set of problems which are (by today's knowledge) not efficiently solvable (for details, see \cite{Garey}). Evidently, if it is not feasible to find max-norm approximations for matrices, it has a fortiori to be true for continuous kernels (which then cannot be  optimally approximated by a separable kernel even on a grid with reasonable effort). This insight suggests to use different, suboptimal, but computationally feasible approximation methods. In the following, we discuss three procedures which have been applied previously in the study of separability for $L^2$-functions.

    \subsubsection{Partial trace approximations}
    \label{Subsubsect_PartTrace}
    
    \noindent Let $A \in \mathcal{ C }((K_1 \times K_2)^2)$ be a kernel function. We can then define the marginal kernels $A_1^{\rm tr } \in \mathcal{ C }(K_1^2)$ and $A_2^{\rm tr } \in \mathcal{ C }(K_2^2) $ point-wise as
\begin{equation} \label{Eq_def_partial_trace_factors}
        A_1^{\rm tr }( s, s' )
    := 
        \int_{ K_2 }  A( s, w, s', w ) dw 
    \quad 
        \textnormal{ and } 
    \quad 
        A_2^{\rm tr }( t, t' ) 
    =  
        \int_{ K_1 }  A( u, t, u, t' ) du 
\end{equation}
    and therewith the separable approximation $A^{\rm tr }$ of $A$ as
    \begin{equation} \label{Eq_TraceApprox}
        A^{\rm tr }( s, t, s', t' ) 
    := 
        \frac{ A_1^{\rm tr }( s, s' ) A_2^{\rm tr }( t, t' ) 
        }{
        \int_{K_1} \int_{K_2} A( u, w, u, w ) d u\, d w
        }.
\end{equation}
    Notice that $A^{ \rm tr  }$ is well-defined, if the denominator in \eqref{Eq_TraceApprox} is non-zero. For a covariance kernel (symmetric and positive semidefinite) this is automatically satisfied, unless 
    the kernel   is trivial. Partial trace approximations are well-known from quantum physics (see \cite{Bhat03} and references therein) and have recently been applied for separability tests of space-time processes in  $L^2$-spaces  (see \cite{ConKokRei17}, \cite{AstPigTav17}). In a recent work of \cite{MasSarPan20} generalizations of partial traces have been investigated in the context of ``almost separable matrices''.

    \subsubsection{Partial product approximations} 
    \label{Subsubsec_PartProd}
    
    \noindent Closely related to the notion of partial traces are partial products.  The partial product approximation depends on a user determined function $\psi \in \mathcal{ C }(K_2^2)$, where typical choices are discussed in \cite{BagDet20} (one being the constant $1$). We then define  for $A \in \mathcal{ C }((K_1 \times K_2)^2)$ point-wise the marginal kernels 
    \begin{equation} \label{Eq_def_partial_prod_factors}
        A_1^{ \rm pr }( s, s' ) 
    := 
        \int_{ K_2^2 }  A( s, w, s', w' ) \psi( w, w' ) dw \, dw' 
    \quad 
        \textnormal{and} 
    \quad 
        A_2^{ \rm pr }( t, t' ) 
    :=  
        \int_{ K_1^2 }  A( u, t, u', t') A_1^{ \rm pr }( u, u' ) du \, du'
    \end{equation}
    and therewith the separable approximation $A^{ \rm pr }$ as   
\begin{equation} \label{Eq_ProdApprox}
        A^{ \rm pr }( s, t, s', t' ) 
    := 
        \frac{ A_1^{ \rm pr }( s, s' ) A_2^{ \rm pr }( t, t' ) 
        }{
        \int_{ K_1^2 } \left( A_1^{ \rm pr }( u, u' ) \right)^2 d u \, d u'
        }.
\end{equation}
    As for partial traces, the approximation is well-defined for a non-vanishing numerator, i.e., for
    \begin{equation} \label{Eq_well_def_prod}
    A_1^{ \rm pr } \neq 0. 
\end{equation}
This condition is fulfilled if $A \neq 0$ for an appropriate the choice of $\psi$. Recently, partial products have been used in \cite{DetDieKut21} for the quantification of separability in  $L^2$-spaces and more recently for the efficient derivation of optimal separable approximations w.r.t.\ the $L^2$-norm in \cite{MasSarPan20}. These latter approximations are discussed next.

    \subsubsection{SPCA approximations} 
    %\label{Subsubsec_SepaPCA}

    \noindent SPCA  (separable principal component analysis) provides  the last approximation method that we want to investigate. With respect to the $L^2$-norm, SPCA-approximations are optimal, whereas for the sup-norm they do not occupy this special position (see Remark~\ref{Lem_NP-Comp-of-MiniProb}). SPCA-approximations have been known in finite dimensions for several decades (\cite{VLoPit93,Gen07}) and have recently been generalized to infinite-dimensional Hilbert spaces by  \cite{DetDieKut21}. The name ``SPCA'' has been introduced by \cite{MasSarPan20}, who proposed an efficient algorithm for the calculation of these approximations by virtue of the partial product (see Section~\ref{Subsubsec_PartProd} in this reference).
    Consider the kernels
$$
        \tilde A_1^{\PCA}(s, s', \bar s, \bar s') 
    := 
        \int_{ K_2^2 } 
        A( s, w, s', w' )
        A( \bar s, w', \bar s', w ) 
        dw \, dw'
$$
    and
$$
        \tilde A_2^{\PCA}( t, t', \bar t, \bar t') 
    := 
        \int_{ K_1^2 } 
        A( u, t, u', t' ) A( u', \bar t, u, \bar t' ) 
        du \, du',
$$
    both of which are  continuous, symmetric (w.r.t.\ to the first and second pair of components) and positive definite. According to Mercer's theorem (Theorem~4.49 in \cite{Steinwart2008}) we can write them down as
    $$
    \tilde A_1^{\PCA}(s, s', \bar s, \bar s')  = \sum_{i \ge 1} \lambda_i v_i(s,s') v_i(\bar s, \bar s'), \quad  \quad \tilde A_2^{\PCA}(t, t', \bar t, \bar t')  = \sum_{i \ge 1} \lambda_i u_i(t,t') u_i(\bar t, \bar t').
    $$
     Here $\{\lambda_i, v_i\}_{i \in \N}$ and $\{\lambda_i, u_i\}_{i \in \N}$ are the eigensystems of the respective integral operators and the eigenvalues are supposed to be in  descending order. It is not difficult to show that if the strict inequality 
    \begin{equation} \label{Eq_well_def_pca}
    \lambda_1 > \lambda_2
    \end{equation}holds, the eigenfunction $v_1, u_1$ are well-defined (up to sign) and continuous (see for both results Appendix~\ref{Subsec_AppCont-and-Diff-SepaPCA}). Supposing that this is true, we  define the marginals
    $$
    A_1^{\PCA}(s,s') = \sqrt{\lambda_1} v_1(s,s'), \quad \quad A_2^{\PCA}(t,t') = \sqrt{\lambda_1} u_1(t,t')
    $$
    and therewith the SPCA approximation
\begin{equation}         \label{Eq_SepPCA-Approx}
        A^{\PCA}(s,t,s',t') 
    = 
        A_1^{\PCA}(s,s') A_2^{\PCA}(t,t'). \\[3ex]
\end{equation}

\noindent We conclude this section with a general result regarding the approximation maps $A \mapsto A^\mathbf{x}$, for every $\mathbf{x} \in \{\rm tr ,  \rm pr , \PCA\}$. We demonstrate that each of these maps is well-defined (in the sense that the resulting kernels are indeed continuous and positive definite). Moreover, we conclude that the approximation maps are Fr\'{e}chet-differentiable, which is critical for the subsequent application of the functional Delta-method (see, for instance, Section~3.9 in \cite{vdVWelBook96}).

\begin{theo} \label{Theo_Frechet}
    Let $K_1 \subset \R^p$, $K_2 \subset \R^q$ be compact,  non-empty sets and let $\tilde A \in \mathcal{ C }((K_1 \times K_2)^2) $ be a covariance kernel with $\tilde A \neq 0$. Moreover, suppose that equation \eqref{Eq_well_def_prod} and \eqref{Eq_well_def_pca} are satisfied. Then the maps
$$
\begin{cases}
&
        \mathbf{ F }_{ i }^{ \mathbf{ x } }: 
            \mathcal{ C }((K_1 \times K_2)^2) 
            \to
            \mathcal{ C }(K_i^2):
            A \mapsto  A_i^\mathbf{x} 
\\
&
        \mathbf{ F }^{ \mathbf{ x } }: 
            \mathcal{ C }((K_1 \times K_2)^2) 
            \to 
            \mathcal{ C }((K_1 \times K_2)^2) :
            A \mapsto  A^\mathbf{x} 
\end{cases}
$$
are for $i=1,2$ and $\mathbf{x} \in \{\rm tr ,  \rm pr , \PCA\}$ well-defined in a sufficiently small, open neighborhood of $\tilde A$ and Fr\'{e}chet differentiable in $\tilde A$. Moreover, 
$
    \mathbf{ F }^{ \mathbf{ x } }_{ i } [ \tilde{ A } ]
$ 
and
$
\mathbf{ F }^{ \mathbf{ x } } [ \tilde{ A } ] 
$ 
are again covariance kernels.
\end{theo}

\noindent The proof of this Theorem can be found in Appendix~\ref{Sec_AppCont-and-DiffSepaMaps}, where we also state the explicit form of the derivatives.
  In the next section, we will use this result in the context of statistical inference for spatio-temporal data.

\section{Testing separability for a continuous covariance kernel} 
\label{Sec_StatInfContCovKer}

\noindent In this section, we develop a statistical test for the hypothesis of a separable covariance kernel in the space of continuous functions. First, we specify the statistical framework in Section~\ref{SubSec_HypoSepaCov} and, second,  investigate test statistics for the hypothesis of separability in Section~\ref{SubSec_AsympTestSepa}. Lemma~\ref{Lem_measure_convergence} entails weak convergence of these test statistics to suprema of Gaussian processes (under the null hypothesis) and hence provides the theoretical tools for separability tests. To approximate the asymptotic quantiles of the test statistics, we present in Section~\ref{SubSec_MultiBootstrap-for-SepaMeas} a multiplier bootstrap for dependent data.

\subsection{Notations and assumptions}\label{SubSec_HypoSepaCov}

\noindent Let $K_1 \subset \R^p$ and  $K_2 \subset \R^q$ be compact,  non-empty sets and 
$( X_n )_{ n \in \mathbb{ Z } }$ be a time series  of random functions in the Banach space $\mathcal{ C }(K_1 \times K_2)$ (for a definition, see Section~\ref{Subsec_SpaceContFunc}). In the following, we will assume that 
$ ( X_n )_{ n \in \mathbb{ Z } }$ 
satisfies \textit{fourth order stationarity}, in the sense that for any  indices $( n_1, \ldots , n_4) \in \mathbb{ Z }^4$ and $k \in \mathbb{ Z }$ the vectors
$
( X_{ n_1 }, \ldots, X_{ n_4 } ) 
$
and
$
( X_{ n_1 + k }, \ldots, X_{ n_4 + k } ) 
$ 
have the same distribution. In particular, supposing that $\mathbb{ E } \|X_1\|^2 < \infty$, both the mean function 
$
    \mathbb{E} X_n ( s, t )$ and the covariance 
\begin{equation*} %\label{Eq_def_covariance_kernel}
        C( s, t, s', t' ) 
    := 
        \mathbb{ E } \left[  X_n ( s, t ) X_n ( s', t' ) \right] 
        -
        \mathbb{ E } \left[ X_n ( s, t ) \right] \, 
        \mathbb{ E } \left[ X_n ( s', t' ) \right]
\end{equation*}
 do not depend on $n$ and are well-defined on the spaces $\mathcal{ C }( K_1 \times K_2 )$ and $\mathcal{ C }((K_1 \times K_2 )^2)$ respectively. In the following, we want to construct a test for the hypothesis of a separable covariance operator, i.e.,
 \begin{equation} \label{Eq_DefH_0}
       H_0:  C \textnormal{ is separable}
    \qquad \textnormal{vs.} \qquad 
       H_1:  C \textnormal{ is not separable},
\end{equation}
where separability is defined in Section~\ref{Subsec_SpaceContFunc}.
For this purpose, suppose that we observe a sample of $N$ random functions $X_1, \ldots, X_N$ (from the time series $( X_n )_{n \in \Z} $). 
For estimating $C$ we use the (standard) empirical covariance estimator defined as:
\begin{equation} \label{def_hat_c}
    \hat{ C }_N  ( s, t, s', t' ) 
  := 
    \frac{ 1 }{ N } 
    \sum_{ n = 1 }^N 
        \left( 
            X_n( s, t )-\overline{ X }_N( s, t ) 
        \right) \,
        \left(
            X_n( s', t' ) -\overline{ X }_N( s', t' ) 
        \right),   
\end{equation}
where 
$
    \overline{ X }_N (s, t )
:= 
    \frac{ 1 }{ N } \sum_{ n = 1 }^N X_n (s, t )
$, $(s, t ) \in K_1 \times K_2 $, denotes the sample mean estimator function. In the next section, we will construct a test statistic for the hypothesis $H_0$, by comparing $\hat C_N$ with a separable approximation. The use of this statistic  is motivated by the fact that, under suitable assumptions on the dependence structure,  $\hat{ C }_N  $ is  a consistent estimator of   $C$ by the law of large numbers (on Banach spaces). In order to quantify dependence, we introduce the popular concept of 
%$\phi$-mixing 
$\alpha$-mixing sequences.

\begin{defi}  \label{Def_Mixing}
Let
 $(X_{n})_{n \in \Z  }$ be a sequence of  random variables on some Banach space.
 For index sets $\mathcal{I}, \mathcal{J}\subset \Z$ we define the  set distance 
 $$
    \operatorname{dist}(\mathcal{I}, \mathcal{J}) 
= 
    \min \{ |i-j|: i \in \mathcal{I}, j \in \mathcal{J}
    \}
 $$
 and use the notation 
 $
 \mathcal{F}_\mathcal{I} 
 :=
 \sigma \left( \{ X_{ i } : i \in \mathcal{I} \} \right)$ 
 for the $\sigma-$algebra generated by the family of random variables  $\{ X_{ i }: i \in \mathcal{I}\} $.
For $r \in \mathbb{N}_0$ the $r$-th $\alpha$-mixing coefficient  is then
defined as
\begin{align*}
    \alpha(r)
=
    \sup 
    &
    \left\{ 
        |\mathbb{P}(A \cap B) 
        - 
        \mathbb{P}(A) \mathbb{P}(B)| 
        \,: \, 
        A \in \mathcal{F}_{ \mathcal{I} }, 
        ~
        B \in \mathcal{F}_{ \mathcal{J} },
        ~
        \operatorname{dist}(\mathcal{I}, \mathcal{J})\ge r,
        ~ \mathcal{I}, \mathcal{J} 
        \subset \Z
    \right\}.
\end{align*}
The sequence $(X_{n})_{n\in \mathbb { Z }  }$ 
 is called   
 %$\phi$-\textbf{mixing},
 $\alpha$-mixing if $\alpha(r) \to 0$, as $r \to \infty$.
\end{defi}

\noindent  We can now state the theoretical  assumptions for the separability test, developed in this section.

\begin{assumption}\label{Assum_MainAssumBanachCLT} $ $
\begin{itemize}
\item[i)] The sequence $(X_n )_{ n \in \Z }$ consists of centered, random functions in
$
\mathcal{ C }(K_1 \times K_2)
$ 
and is fourth order stationary.
\item[ii)] There exists a non-negative random variable $M$, parameters $\beta \in \left( 0, 1 \right]$ and $J \ge 0$ with the constraint
$J \beta > \lceil 2  ( \operatorname{dim}(K_1) + \operatorname{dim}(K_2) ) \rceil + 1$ such that 
\begin{equation} \label{Eq_HoelderContData}
    \mathbb{ E } ( \| X_1 \|^J M^{ J } ) 
    < 
    \infty,
\quad
\text{ and }
\quad
		| X_n ( s, t ) - X_n ( s', t' ) |
	\le
		M 
	%	|  ( s, t ) - ( s', t' ) |^{ \beta }
    \max\{  \| s - s' \|^{ \beta }, \| t - t' \|^{ \beta }  \},
\end{equation} 
	holds (almost surely) for all $( s, t ) , ( s', t' )  \in K_1 \times K_2$ 
	and all $n \in \mathbb { Z }$. 
\item[iii)] For some $\gamma > \max \{ J, 8 \}$ it holds that
$      \mathbb{ E } \|  X_1 \|^{ \gamma } 
    < 
        \infty.
$
\item[iv)] The sequence $(X_n)_{n \in \Z}$ is $\alpha$-mixing in the sense of Definition~\eqref{Def_Mixing}, with $\alpha(r) \le \const /  ( 1 + r )^{ a }$, where $\kappa>0$ is some constant and $a > 2 \gamma / (\gamma - 8 )$. 

\end{itemize}
\end{assumption}

\noindent We briefly discuss each of these Assumptions.

\begin{rem} \mbox{} 
\rm
\begin{itemize}
    \item[$i)$] We require second order stationarity s.t.\ the empirical mean $\overline{ X}_N$ and the empirical covariance operator $\hat{ C }_N $ are consistent estimators. The stronger assumption of fourth order stationarity guarantees existence of the long-run variance operator of $\sqrt{N}(\hat{ C }_N - C )$. An examination of our proofs shows that these assumptions can be further relaxed to  conditions on the moments of $X_n$. However, for  parsimony of presentation, we do not discuss these mathematically weaker (but harder to understand) adaptions. Our stationarity assumption is weaker than those in the related literature (both in $L^2$- and $\mathcal{ C }$-spaces), where  either independence or strict stationarity is considered (see \cite{ConKokRei17,AstPigTav17,BagDet20,DetKok21}).
    \item[$ii)$] To devise asymptotic tests for separability, we derive a CLT on the Banach space of continuous functions. Such results require the validation of tightness conditions, which depend on the geometry of the underlying space (reflected by the entropy rate). 
    In the case of i.i.d.\ observations, sufficient conditions for a CLT can be found in \cite{JainMar75}, and for the dependent   case (under $\alpha$- and 
    $\phi$-mixing) in \cite{DmitErmOst84}. 
    More specifically, Theorem~1 of \cite{JainMar75} requires a random Lipschitz condition, which
    together with an entropy condition entails a functional CLT on $\mathcal{ C } (K)$.   
    In this paper, we further relax this assumption by  requiring only a  random $\beta$-H\"{o}lder condition. 
    This condition specifically includes sample paths of the Brownian motion, that are almost surely $\beta$-H\"older continuous, for $\beta < 1/2$. 
    However, additional smoothness helps to reduce moment conditions imposed on the data (see the next assumption).

    \item[$iii)-iv)$] The existence of sufficiently many moments is the  key to proving our CLT (and its bootstrap variant). As might be expected, stronger moment conditions can be traded off against weaker smoothness assumptions for the data functions, as well as milder conditions on temporal dependence. We quantify dependence by the decay rate of strong mixing coefficients, where a slower decay expresses stronger time-dependence. In comparison to the related literature, our mixing conditions are rather weak and include large classes of dependent time series.

\end{itemize}
    
\end{rem}

\subsection[Central limit theorems in $\mathcal{ C }((K_1 \times K_2)^2 )$]{Central limit theorems in $\boldsymbol{ \mathcal{ C }((K_1 \times K_2)^2 ) }$}\label{SubSec_AsympTestSepa}

\noindent We begin our theoretical derivations by proving weak convergence of the standardized empirical covariance operator $\sqrt{N}(\hat{ C }_N  - C)$ to a Gaussian process $G$. This result, together with the corresponding bootstrap convergence, is of independent interest for the statistical investigation of the covariance on the space of continuous functions (notice that Theorem  \ref{Theo_CLT-Rv-and-EmpCovOp} is a special consequence of Theorem \ref{Theo_MultiplierBootstrap}).

\begin{theo} \label{Theo_CLT-Rv-and-EmpCovOp}
    Suppose that Assumption~\ref{Assum_MainAssumBanachCLT} holds. 
    	Then there exists a centered Gaussian process $G$ on the space $\mathcal{ C }((K_1 \times K_2)^2)$ such that 
$$
	      \sqrt{ N }  \big(\hat{ C }_N  - C \big)
    \stackrel{d}{\to} 
        G.
$$

\end{theo}

\noindent We now combine weak convergence of the empirical covariance with the Fr\'{e}chet-differentiability of the separable approximation maps (Theorem~\ref{Theo_Frechet}). This yields weak convergence of the empirical separability measure 
$        \|
            \hat{ C }_N  
            - 
             \hat{ C }_N^{ \mathbf{ x } }
        \| 
$ 
for any one of the approximation types $\mathbf{x} \in \{ \rm tr ,  \rm pr , \PCA \}$, considered in Sections~\ref{Subsec_AppCont-and-Diff-PartProd}-\ref{Subsec_AppCont-and-Diff-SepaPCA}.

\begin{lem} \label{Lem_measure_convergence}
    Suppose our Assumption~\ref{Assum_MainAssumBanachCLT} holds, that $C \neq 0$ is separable and  that \eqref{Eq_well_def_prod}, \eqref{Eq_well_def_pca} are satisfied. Then under the hypothesis $H_0$ (defined in \eqref{Eq_DefH_0}, it holds that
\begin{equation} \label{Eq_conv_approx}
         \sqrt{ N }  \left( \hat{ C }_N  
            - 
             \hat{ C }_N^{ \mathbf{ x } }
        \right) 
    \stackrel{d}{ \to}
         (\operatorname{ Id } - D_C  \mathbf{ F }^{ \mathbf{ x } } ) G
\end{equation}
    for all $\mathbf{x} \in \{ \rm tr ,  \rm pr , \PCA \}$. Here,  $D_C  \mathbf{ F }^{ \mathbf{ x } } $ is the  derivative of the separable approximation map (see Section~\ref{Sec_MathConc}),  $\operatorname{ Id }$ the identity operator and $ G $ the Gaussian process from Theorem~\ref{Theo_CLT-Rv-and-EmpCovOp}. In particular, it holds under $H_0$ (separability), that
\begin{equation} \label{Eq_conv_H_0}
         \sqrt{ N }  
        \left\|
            \hat{ C }_N  
            - 
             \hat{ C }_N^{ \mathbf{ x } }
        \right\| 
    \stackrel{d}{ \to}
         \left\| (\operatorname{ Id } - D_C\mathbf{ F }^{ \mathbf{ x } } ) [ G ] 
         \right\|
    %    \| G - D A^\mathbf{x} ( C ) [ \mathcal{ G } ]\|,
\end{equation}
and under $H_1$ (non-separability) that
\begin{equation*} %\label{Eq_consistency}
         \sqrt{ N }  
        \left\|
            \hat{ C }_N  
            - 
         \hat{ C }_N^{ \mathbf{ x } }
        \right\| 
    \stackrel{d}{ \to}
        \infty.
\end{equation*}
\end{lem}

\noindent The weak convergence in  \eqref{Eq_conv_approx} follows by a version of the Delta-method (see for instance Section~3.9 in \cite{vdVWelBook96}) applied to the process $\sqrt{N}( \hat{ C }_N  - C )$. 
The statement in \eqref{Eq_conv_H_0} is a direct consequence of this, using the continuous mapping theorem (see Theorem~1.3.6 in \cite{vdVWelBook96}). Finally, divergence under $H_1$ is entailed by consistency of $\hat{ C }_N $ and 
$ 
    \hat{ C }^{\mathbf{x} }_N.
$ 

\noindent Lemma~\ref{Lem_measure_convergence} implies a straightforward test for the hypothesis $H_0$: If the Gaussian process $ G $ (or equivalently its covariance) were known, we could easily approximate the upper $(1-\alpha)$ quantile $q_{ 1 - \alpha}$ of the limiting distribution and reject whenever
$$
         \sqrt{ N }  
        \left\| 
            \hat{ C }_N  - \hat{ C }_N ^\mathbf{x}
        \right\| 
    > 
        q_{ 1 - \alpha }.
$$
    Unfortunately, the covariance $C_G$ of $G$ is unknown and extremely difficult to approximate. Notice that it is defined  on the product space $\mathcal{ C }( ( K_1 \times K_2 )^4)$, making it nearly impossible: either to calculate or save (recall that our discussion is motivated by the intractability of $C$, and  $C_G$ consumes the squared amount of memory space). In the literature on $L^2$-tests, \cite{AstPigTav17,ConKokRei18} have tried to evade this problem by using projection methods. While this approach is viable for functional data with a discrete spatial component, it seems less effective in the case of continuous spatial data (see also our simulations in Section~\ref{Subsec_simulations}).
    As a more practicable and fully functional alternative, we therefore discuss a  multiplier bootstrap for dependent data in the next section.

    \subsection{A multiplier bootstrap}
    \label{SubSec_MultiBootstrap-for-SepaMeas}

\noindent In order to approximate the  distribution of $ \sqrt{ N }  (\hat{ C }_N  - C )$ we propose a multiplier bootstrap for dependent data, which is inspired by the methodology of \cite{BucKoj13}
and  has recently been adapted to functional data by \cite{DetKokAue20}. While various alternative bootstraps exist for  functional data (an important example for the covariance is \cite{PapSap2016}), their validity is usually demonstrated in an $L^2$-framework, making them inapplicable in our setup. \\
We begin our discussion by fixing a number $r \in \N$ of bootstrap replicates. For $1 \le k \le r$, we consider vectors of random weights $(w_{1,N}^{(k)},\ldots,w_{N,N}^{(k)}) \in \R^N$, which are independent of the data $X_1, \ldots, X_N$ and independent across $k$. We assume that each weight-vector follows a multivariate normal distribution. More precisely, the variables $w_{ i , N }^{ ( k ) }$ are supposed to be centered with unit variance, and $(l_N - 1 )$-dependent with 
$$
        \mathbb{ E } [ w_{ i, N }^{ ( k ) } w_{ j, N }^{ ( k ) } ]
    =
        ( 1 - | i - j | / l_N )
$$ for any $|i-j| \le l_N$.
Here $l_N \in \N$ is a bandwidth parameter, comparable to the block-length in a block bootstrap (see, for instance, Theorem~2.1 in \cite{BucKoj13}). We then define the bootstrapped process 
$ B_N^{ ( k ) }$ point-wise as 
\begin{equation} \label{Eq_boots_var}
        B_N^{ ( k ) }(s,t,s',t')
    := 
        \frac{ 1 
        }{ 
        N 
        } 
        \sum_{n=1}^N  w_{n,N}^{(k)} 
        \,
        \left\{ 
            [ X_n ( s, t ) 
            - 
            \overline{ X }_N (s,t) 
            ][ 
            X_n(s',t')  
            - 
            \overline{ X }_N ( s', t' ) ] - 
            \hat{ C }_N (s,t,s',t')
        \right\}
\end{equation}
for $1 \le k \le r$.

\begin{theo}[Multiplier Bootstrap]\label{Theo_MultiplierBootstrap}
	Suppose that Assumption~\ref{Assum_MainAssumBanachCLT} holds and that $l_N$ satisfies both restrictions $l_N \to \infty$ and  $l_N/ \sqrt{ N }  \to 0$. Then for any $r \in \N$ the weak convergence
\begin{equation}  \label{Eq_weak_con_bootstrap}
       \sqrt{ N }  
       \left(
         \hat{ C }_N  - C  ,  B_N^{ ( 1 ) }, \ldots, 
             B _N^{ ( r ) } 
       \right)  
    \stackrel{ d }{ \to }
        \left( 
            G, G^{ ( 1 ) }, \ldots, G^{ ( r ) } 
        \right)
\end{equation} 
	holds, where $G$ is the Gaussian process from Theorem~\ref{Theo_CLT-Rv-and-EmpCovOp} and $G^{ ( 1 ) }, \ldots, G^{ ( r ) }$ are i.i.d.\ copies of $G$.
\end{theo}

\noindent 
A detailed proof of  Theorem~\ref{Theo_MultiplierBootstrap} is given in the Appendix, but let us briefly  give some arguments  why \eqref{Eq_weak_con_bootstrap} holds. 
First, notice that due to independence of the weights (of the data and among themselves across $k$), we have uncorrelatedness of $B_N^{(k)}$ with $\hat{ C }_N -C$ and with any other $B_N^{(k')}$. Hence, if these objects are (jointly) normal in the limit, they also have to be independent. This leaves open the questions, why $B_N^{(k)}$ is normal and why it has the right covariance structure. On a high level, asymptotic normality of $B_N^{(k)}$ is obvious, as it is a sum of weakly dependent random variables. A closer look reveals that dependence in  \eqref{Eq_boots_var} is governed by the bandwidth $l_N$, which cannot increase too fast for a CLT to hold (this will be reflected in the assumption $l_N = o(\sqrt{N})$). On the other hand, $l_N$ has to diverge for $N \to \infty$, s.t.\ $B_N^{(k)}$ has the same asymptotic covariance as $G$. As $l_N$ grows, neighboring terms in $B_N^{(k)}$ have (almost) identical weights. Consequently, their covariance is (almost) identical to that of the terms in $\hat{ C }_N $, without a multiplier. 

\noindent
By Theorem~\ref{Theo_MultiplierBootstrap} 
the bootstrap variable $B_N^{ ( k )}$ mimics the random fluctuations of the empirical covariance $\hat{ C }_N $ around $C$. Thus, it can help us to approximate the distribution of the separability measure
$ 
        \|
            \hat{ C }_N  
            - 
          \hat{ C }_N ^{ \mathbf{ x } }
        \|;
$
the actual object of interest.
To make this clear, we  use the representation 
$$ 
        \|
            \hat{ C }_N  
            - 
             \hat{ C }_N ^{ \mathbf{ x } }
        \| 
    = 
        \|
            \{\hat{ C }_N - C \} 
            - (
            [ C + \{ \hat{ C }_N  - C \} ]^{ \mathbf{ x } }
            - C)
        \|.
$$
To give a bootstrap approximation, we can replace $\{\hat{ C }_N -C\} $ by $ B_N^{ ( k )}$. Furthermore, since we want to approximate the distribution under the null, we have to replace $C$ by the separable estimator $\hat{ C }_N ^{\mathbf{ x } }$ everywhere else. This gives us the bootstrap version
\begin{equation} \label{Eq_def_bootstrap_statistic}
        B_N^{ \infty, ( k ) }
    := 
        \|
        B_N^{ ( k )} 
        - 
        (
        [\hat{ C }_N ^{\mathbf{ x } } + B_N^{ ( k )}]^{ \mathbf{ x } }
        -
        \hat{ C }_N ^{\mathbf{ x } } )
        \|.
\end{equation}
\begin{lem} \label{Lem_consistency_boots}
    Under the assumptions of Theorem~\ref{Theo_MultiplierBootstrap} it holds that
    \begin{align*}
      \sqrt{N}(B_N^{\infty,  ( 1 )}, \cdots, B_N^{\infty,  ( r )})  \stackrel{d}{ \to}
         \big(\left\| (\operatorname{ Id } - D_C\mathbf{ F }^{ \mathbf{ x } } ) [ G_1 ] 
     \right\|,\cdots,\left\| (\operatorname{ Id } - D_C\mathbf{ F }^{ \mathbf{ x } } ) [ G_r ] 
         \right\| \big)~,
    \end{align*}
     where $G$ is the Gaussian process from Theorem~\ref{Theo_CLT-Rv-and-EmpCovOp} and $G^{ ( 1 ) }, \ldots, G^{ ( r ) }$ are i.i.d.\ copies of $G$.
\end{lem}
\noindent Lemma~\ref{Lem_consistency_boots} theoretically underpins the following test decision: Let
 $\hat q_{1-\alpha}^{(r)}$ be the empirical $(1-\alpha)$ quantile of the vector $( B_N^{ \infty, ( 1 )},\cdots, B_N^{\infty,  ( r )})$. Reject the hypothesis $H_0$ (defined in \eqref{Eq_DefH_0}), if 
\begin{equation} \label{Eq_test_decision}
    \| \hat{ C }_N   - \hat{ C }_N ^{ \mathbf{ x } }\|>\hat q_{1-\alpha}^{(r)}.
\end{equation}
Lemma~\ref{Lem_consistency_boots} implies that the resulting test holds asymptotic level $\alpha$, as $r \to \infty$, that is 
$$
    \lim_{r \to \infty } 
    \lim_{N \to \infty }
    \mathbb{ P }_{H_0}\left\{ 
        \| \hat{ C }_N   
        - 
        \hat{ C }_N ^{ \mathbf{ x } }\| 
        >
        \hat q_{1-\alpha}^{(r)}
    \right\}
= 
    \alpha,
$$
and is consistent as $N \to \infty$ under the alternative $H_1$, that is 
$$
    \lim_{N \to \infty }
    \mathbb{ P }_{H_0}\left\{ 
        \| \hat{ C }_N   
        - 
        \hat{ C }_N ^{ \mathbf{ x } }\| 
        >
        \hat q_{1-\alpha}^{(r)}
    \right\}
= 
1
$$
for any $r \in \mathbb{N}$.

\begin{rem} \label{Rem_RemarkOverBootstrap}\rm \mbox{}
\begin{itemize}
\item[$i)$] There are different ways to mathematically validate a bootstrap procedure. One way is proving convergence of the bootstrap measure (conditional on the data) to the correct limiting distribution. Here, convergence is measured w.r.t.\ some metric on the probability measures, such as the Kolmogorov--Smirnov distance for probability measures on $\mathbb{R}^d$ or the bounded Lipschitz distance on general metric spaces (see \cite{vdVWelBook96} for a definition). This approach of studying conditional distributions is considered in many classical works, such as \cite{Hall1992}. As an alternative, it is possible to derive unconditional convergence results (such as Theorem~\ref{Theo_MultiplierBootstrap} and Lemma~\ref{Lem_consistency_boots}), where it is shown that a number of bootstrapped statistics converge to independent copies of the same limiting distribution. Such results reflect the  need of generating bootstrap repetitions for most practical test decisions. One merit of this approach is its clearer interpretability compared to the abstract notion of convergence on the spaces of measures, w.r.t.\ a difficult metric. Yet, from a theoretical standpoint, the two approaches are in many instances equivalent. In particular, Lemma~\ref{Lem_consistency_boots} implies for the conditional bootstrap measure $\mathbb{P}^{B_N^{(1)}|X_1, \cdots,X_N}$ the convergence in probability
$$
        d \left( \mathbb{P}^{B_N^{(1)}|X_1, \cdots,X_N}, \mathbb{P}^{\|
            \hat{ C }_N  
            - 
             \hat{ C }_N ^{ \mathbf{ x } }\|
        } \right) 
    \overset{\mathbb{P}}{\to} 
        0,
$$
where 
$d$ is the Kolmogorov--Smirnov  metric for probability measures on $\R$ and 
$\mathbb{P}^{\|
            \hat{ C }_N  
            - 
             \hat{ C }_N ^{ \mathbf{ x } }\|}$ denotes the probability
             measure of the test statistic $\|
            \hat{ C }_N  
            - 
             \hat{ C }_N ^{ \mathbf{ x } }\|$. The proof of this assertion follows directly from Lemma~2.3 in \cite{Buecher2019} and a similar, conditional version of Theorem~\ref{Theo_MultiplierBootstrap} can be given by using the techniques in Section~3 of that paper.

    \item[$ii)$] The bandwidth parameter $l=l_N$ in Theorem~\ref{Theo_MultiplierBootstrap} plays a similar role as the block length in a block bootstrap. In the special case of independent data it is not necessary to assume that $l_N \to \infty$ and a choice $l=1$ yields the desired result (for multiplier bootstraps in the independent case, see also Section~3.6 in \cite{vdVWelBook96}). We also want to highlight that our assumption $l_N/ \sqrt{ N }  \to 0$ is weaker than in previous works, where usually a polynomial decay rate of $l_N/\sqrt{N}$ was assumed 
        (see \cite{DetKok21}). 
        \item[$iii)$] The bootstrap test decision presented in this paper can be implemented without ever saving the entire empirical  covariance operator, or any other object of comparable size. In the following, we want to sketch an argument why this is true. Let us therefore focus on the partial trace approximation: To approximate the partial trace approximation $(\hat{ C }_N)^{tr}$, it suffices to calculate the trace of the empirical covariance $\operatorname{ Tr }[\hat{ C }_N]$, as well as the partial traces $(\hat{ C }_N)_1^{tr}$, $(\hat{ C }_N)_2^{tr}$. All of these objects can be calculated directly from the data, as e.g.,
        \begin{equation} \label{eq_calc_traces}
        \operatorname{ Tr } [ \hat{ C }_N]
    =
        \frac{ 1 }{ N } \sum_{ i = 1 }^N
        \int_{K_1}\int_{K_2}
        X_i^2 (s, t ) ds \, dt
\end{equation}
or
\begin{equation} \label{eq_calc_patraces}
        (\hat{ C }_N)_1^{tr} ( s, s' )
    =
        \frac{ 1 }{ N } \sum_{ i = 1 }^N
        \int_{K_2}
        X_i( s, t ) X_i( s', t ) dt.
\end{equation}
In particular, we do not have to save $\hat{ C }_N$ to calculate  $(\hat{ C }_N)^{tr}$. In contrast to $\hat{ C }_N$, the three objects $\operatorname{ Tr }[\hat{ C }_N]$, $(\hat{ C }_N)_1^{tr}$, $(\hat{ C }_N)_2^{tr}$ are small (a discretization takes about as much memory as a data function $X_n$) and hence we assume that they (and thereby $(\hat{ C }_N)^{tr}$) are tractable. So, from now on we assume that these objects are saved. Now, calculating the test statistic $\|
            \hat{ C }_N  
            - 
             \hat{ C }_N ^{ tr}\|$ can be done by maximizing the distance $|
            \hat{ C }_N  
            - 
             \hat{ C }_N ^{ tr }|$ blockwise (it is easy to calculate for example $\hat C_N$ only for certain subsets of the arguments $(s,t,s',t')$ and then evaluate the maximum over all subsets). The calculation of the bootstrap statistic  $B_N^{\infty,  ( k )}$ (defined in  \eqref{Eq_def_bootstrap_statistic}) 
             is slightly more intricate. Here, we have to calculate $[\hat{ C }_N ^{tr }
             + B_N^{ ( k )}]^{ tr} $ (the other objects in $ B_N^{\infty,  ( k )}$ can again be calculated  for subsets of indices in a
             straightforward way).
             Calculating $[\hat{ C }_N ^{tr }+ B_N^{ ( k )}]^{ tr} $ boils down to calculating the (partial) traces of $\hat{ C }_N ^{tr }$ and $B_N^{ ( k )}$
             separately (as the partial traces are linear). For $\hat{ C }_N ^{tr }$ they equal $(\hat{ C }_N)_1^{tr}$, $(\hat{ C }_N)_2^{tr}$ and we have already calculated and saved them before. For $B_N^{ ( k )}$, they are equal to the partial traces of $\frac{ 1 
        }{ 
        N 
        } 
        \sum_{n=1}^N  w_{n,N}^{(k)} \hat C_N$ (so essentially those of $\hat C_N$) and those of $ \frac{ 1
        }{ 
        N 
        } 
        \sum_{n=1}^N  w_{n,N}^{(k)} 
        \,
            [ X_n
            - 
            \overline{ X }_N 
            ]\cdot
            [ 
            X_n
            - 
            \overline{ X }_N  ] $. Notice that we cannot save this later object (it is of the same size as $\hat C_N$). So, as a computational trick, we can express it as the “covariance” of the complex valued data $\sqrt{ \operatorname{sign}(w_{n,N}^{(k)}) }\cdot|w_{n,N}^{(k)}| \cdot [X_n
            - 
            \overline{ X }_N] $. These objects are again small enough to be saved, and from them, we can calculate the partial traces of $ \frac{ 1
        }{ 
        N 
        } 
        \sum_{n=1}^N  w_{n,N}^{(k)} 
        \,
            [ X_n
            - 
            \overline{ X }_N 
            ]\cdot
            [ 
            X_n
            - 
            \overline{ X }_N  ] $ directly (see \eqref{eq_calc_traces} and \eqref{eq_calc_patraces}).  Notice that in practice it is  necessary to take the real part in the end, to eliminate complex valued remainders, due to computational imprecisions.
\end{itemize}
\end{rem}

\section{Finite sample properties} \label{Sec_finite_sample}

\noindent In this section, we study the finite sample performance of our new method. We begin by a simulation study, where we compare the performance of the test with a benchmark procedure from \cite{ConKokRei18}. Subsequently, we apply our test to a dataset of roman language recordings as described and used in \cite{AstPigTav17}.

\subsection{Simulations} \label{Subsec_simulations}

\noindent Following \cite{ConKokRei18}, we generate spatio-temporal data by virtue of a functional $\operatorname{MA}(1)$-process. More precisely, we generate Gaussian processes $e_0, \dots,e_{ N  }$, living on the unit square $[0,1]^2$, with covariance kernel $C$, point-wise defined as
\begin{align} \label{Genitings_sigma}
    C(s,t,s',t')
:=
&
    \frac{1}{(a|t-t'|+1)^{1/2}} 
    \exp\left( -\frac{b^2|s-s'|^2}{(a|t-t'|+1)^c}\right),
\qquad 
    s,s',t,t' \in \left[ 0, 1 \right].
\end{align}
As in \cite{ConKokRei18}, we set $a=3$ and $b=2$ and consider the observations 
\begin{equation*}
    X_n(s,t)
:= 
    \sum_{s'=1}^S \exp\left\{ -b^2 (s-s')^2 \right\} [e_n(t, s')+e_{n-1}(t, s')] 
\quad 
    n=1, \cdots, N. %\label{simulated_data}
\end{equation*}
Notice that for $c=0$  the kernel $C$ in \eqref{Genitings_sigma} factorizes into purely spatial and temporal components. By implication, the covariance of $X_n$ is separable for $c=0$ (the hypothesis). On the other hand, if $c>0$ the kernel $C$ is not separable and this inseparability is inherited by the covariance of $X_n$. We hence generate data under $H_0$ by setting $c=0$ and under the alternative by setting $c=1$. As in \cite{ConKokRei18}, we discretize the time component $t$, by dividing the unit-interval into $T = 50$ equidistant points ($1/T, 2/T, \ldots, 50/T$). For the spatial component,  \cite{ConKokRei18} use a similar discretization ($1/S, 2/S, \ldots , (S-1)/S$), but for smaller numbers of gridpoints with 
$S = 4,6,8,10,12,14$. Considering larger $S$ is problematic for the procedure in \cite{ConKokRei18}, which relies on the estimation of the asymptotic covariance operator of $\sqrt{N}(\hat{ C }_N - C )$ - an expensive and difficult undertaking in high dimensions (we touched this point in our discussion at the end of Section~\ref{SubSec_AsympTestSepa}). To implement their procedure, \cite{ConKokRei18} rely on dimension reduction, for $S \le 8$ in time, and for $S>8$ in both space and time. As we might expect, such projections entail deteriorating power for larger $S$ as the amount of variance explained decreases. In our study, we use the same number of gridpoints (to allow meaningful comparisons) but also study the larger sizes of $S=20$ and $S=30$. As sample sizes, we consider $N=50, 100, 150, 200$ and as corresponding block-lengths $l_N = 2,2,3,4$. The number of bootstrap repetitions for the generation of the empirical quantile is fixed at $400$ and the nominal level at $\alpha=5\%$. All reported results are based on 1000 simulation runs. \\ In Table~\ref{Table_1} we report empirical rejection probabilities under the hypothesis of separability ($c=0$) and the alternative ($c=1$). In brackets, we include the rejection probabilities reported in \cite{ConKokRei18} (whenever available), where we have chosen each time the maximum number of projection parameters (in time for $S \le 8$ and in space and time for $S>8$), which produced the most powerful results (see their Tables~I-IV for a full picture of performance under variation of the projection parameters). 

\begin{table}[H]
{\footnotesize
\begin{center}
\begin{tabular}{|*{9}{c|}}
\hline  \multicolumn{1}{|c|}{   }  &
\multicolumn{4}{c|}{ rejection probability under $H_0$} & \multicolumn{4}{|c|}{  rejection probability under $H_1$ } \\  \hline
 $   S$ &   $N=50$ &   $N=100$ &    $N=150$ &   $N=200$&   $N=50$ &    $N=100$ &    $N=150$ &   $N=200$   \\ \hline 
 $4$ & $3.3$ & $5.8$ $(5.0)$ & $5.2$ $(5.3)$ & $5.4 $ $(5.5)$ & $36.0$ & $91.0$ $(95.1)$&$99.2$ $(99.8)$&$100.0$ $(100.0)$\\ \hline 
 $6$ & $2.5$ & $6.0$ $(5.3)$ & $5.2$ $(6.6)$&$4.5 $ $(5.1)$ &$40.5$ & $92.8$ $(89.0)$ &$99.1$ $(99.3)$ &  $100.0$ $(100.0)$\\ \hline 
 $8$ & $3.75$ & $6.5$ $(7.5)$ & $4.2$ $(4.7)$& $5.1 $ $(5.7)$&$34.8$ &$92.3$ $(85.2)$ &$98.0$ $(98.7)$&$99.9$ $(100.0)$\\ \hline 
 $10$ &  $2.75$ &$7.6$ $(4.7)$&$3.4$ $(5.3)$ &$3.7 $ $(5.8)$ & $41.8$& $90.6$ $(84.7)$& $98.6$ $(98.5)$ &$99.9$ $(100.0)$\\ \hline 
 $12$ &$1.5$ & $ 7.2$ $(5.0)$ & $ 3.1$ $(6.0)$&$5.2 $ $(5.7)$ & $37.8$ & $90.4 (82.9)$ & $99.0$ $(97.8)$&$99.9$ $(100.0)$\\ \hline 
 $14$ & $2.0$ & $ 7.0$ $(4.6)$ &$ 5.3$ $(5.7)$ &$3.7 $ $(5.6)$ & $33.5 $ & $92.5$ $(78.9)$ & $98.6$ $(93.7)$&$100.0$ $(95.7)$\\ \hline 
 $20$ & $3.4$ &$5.2$ $\qquad$ & $5.7$ $\qquad$  & $3.80$ $\qquad$ & $38.0$ &$90.3$ $\qquad$ &$98.3$ $\qquad$ &$99.9$ $\qquad$\\ \hline 
 $30$ &$2.7$ &$7.0$ $\qquad$ & $5.3$ $\qquad$& $3.61$ $\qquad$ & $36.1$ &$87.4$ $\qquad$ &$98.2$ $\qquad$ &$99.9$ $\qquad$\\ \hline 
\end{tabular}
$ $\\
\caption{\label{Table_1}\textit{
 Empirical rejection probabilities of the bootstrap test \eqref{Eq_test_decision} under the hypothesis ($c=0$) and the alternative ($c=1$). Benchmark values from the test of \cite{ConKokRei18} are given in brackets if available.}}
\end{center}
}
\end{table}

\noindent Our results in Table~\ref{Table_1} attest a satisfactory performance of the bootstrap test. The nominal level is reasonably approximated for $N \ge 100$, while for $N=50$ the test is somewhat conservative.  The power of the bootstrap test is high in most scenarios. While for small values of $S$, the benchmark test fares slightly better, the bootstrap's performance does not deteriorate for larger $S$, where it clearly outperforms the benchmark. Even raising $S$ to $20$ or $30$ does not impinge 
on performance in any systematic way (exactly what a theory for continuous processes would suggest). Computationally, the bootstrap is particularly user-friendly: It allows a straightforward parallelization in the generation of bootstrap samples and is hence easy to implement. We have run our simulations on a standard desktop computer (3.2 GHz Apple M1 Pro, Octa-core, 16 GB RAM) and any test evaluation needed less than a minute (for $S \le 10$ less than $2$ seconds), which underpins the practicability of this approach.
The subsequent data example was run on the same machine for even larger values of $T$ and $S$.

\subsection{A data example}
\noindent    We apply our test to the \textit{acoustic phonetic} dataset of acoustic (log-)spectograms of \cite{AstPigTav17}. A brief discussion with additional references can be found in  Section~4.2 of their paper. 
    Both the raw and preprocessed data together with detailed descriptions are contained in the file \path{Acoustic_Data_And_Code.zip} available from \href{https://rss.onlinelibrary.wiley.com/hub/journal/14679876/series-c-datasets/67_5}{Series C datasets of Volume 67} (\path{https://rss.onlinelibrary.wiley.com/hub/journal/14679876/series-c-datasets/67_5}). The data consist of recordings of spoken words (in this case the numbers one to ten) in five different roman languages. For statistical use they have been transformed in the form of acoustic log-spectograms. For our purposes only the preprocessed data are used (namely the files \path{WarpedPSD.RData} and \path{SVRF_WarpedPSD_SuppMat.RData}). There are in total 219 data "functions" in a frequency-time domain as described in Sections~2-4 of \cite{PigHadColAst18}.

\noindent As already indicated by the analysis in  \cite{PigHadColAst18} the null hypothesis~\eqref{Eq_DefH_0} of separability seems to be violated in this instance. Before applying our test statistic we look at relative measures of separability. More concretely, for the separable trace approximation, the relative measure is given by 
$$
    \frac{ \| \hat{ C }_N 
        - 
        \mathbf{ F }^{ \mathbf{ tr } } 
        ( \hat{ C }_N ) \|  
    }{
    \| \hat{ C }_N \|
    }.
$$

 \noindent   
Indeed, as a preliminary analysis, we found that the relative measure of separability for the language covariance operators are rather high, 
see Table~\ref{Table_RelMeasSepaLanguages}.

\begin{table}[H]
    \centering
    \begin{tabular}{>{\centering}m{2cm}||c|c|c|c|c }
        & {\bf French } & {\bf Italian } & {\bf Portuguese  } &  {\bf American Spanish } & {\bf Iberian Spanish } \\ \hline
    {\bf Relative measure } & 0.631 & 0.895 & 0.868 & 0.947 & 0.882 \\
    \end{tabular}
    \caption{\it The relative measure of separability w.r.t.\ the trace approximation of the five roman languages.}
    \label{Table_RelMeasSepaLanguages}
\end{table}

\noindent    For every of the five languages we ran our bootstrap statistic of $1,000$ repetitions on the residual (i.e.\ centered) surface data for each language separately. As a result the hypotheses of separability for every of these languages is rejected with a highly significant $p-$value of less than $0.1 \%$. See the first row of Table~\ref{Table_P-valuesDataExample} for the p-values of our test. For the sake of completeness the second and third row contain the p-values obtained in \cite{AstPigTav17} for 2 frequency and 3 time dimensions and 8 frequency and 10 dimensions, respectively.
     Our method provides several advantages, it is free of any tuning parameter, in contrast to the Studentized version of the empirical bootstrap of \cite{AstPigTav17}, who have to choose additional parameters of dimensions of eigendirections (note that their test can  only detect deviations from  
     separability along those eigendirections). 
    Moreover, when too few  eigendirections are chosen (see Section~4.2 of \cite{AstPigTav17}) the hypotheses of separability cannot be rejected at the $5\%$-level. 
    Secondly, in order to calculate the bootstrap we do only need to save the data, not the whole covariance operator, hence avoiding storage problems.
    
\begin{table}[H]
    \centering
    \begin{tabular}{>{\centering}m{2.5cm}|| c | c | c | c | c  }
        & {\bf French } & {\bf Italian } & {\bf Portuguese  } &  {\bf American Spanish } & {\bf Iberian Spanish } \\ \hline
    {Test  \eqref{Eq_test_decision}  } & <0.001 & <0.001 & 0.001 & <0.001 & <0.001 \\
    {Emp.\ stud.\ test }  $(2,3)$ & 0.078 & 0.197 & 0.022 & 0.360 & 0.013 \\
    { Emp.\ stud.\  test } $(8,10)$ & 0.001 & 0.002 & 0.001 & 0.001 & <0.001 \\
    \end{tabular}
    \caption{\it  $p-$values of three different bootstrap tests of five Roman languages. First row: the test  \eqref{Eq_test_decision} proposed in this paper. 
    Second and third row: the  studentized version of the empirical bootstrap test proposed by \cite{AstPigTav17}
 with frequencey and time dimensions $(2,3)$ and $(8,10)$ respectively.}
    \label{Table_P-valuesDataExample}
\end{table}

\bigskip 

\noindent
\textbf{Acknowledgements.} 
This work  was partially supported by the  
 DFG Research unit 5381 {\it Mathematical Statistics in the Information Age}.

\bibliographystyle{chicago} 
\bibliography{main}

\newpage

\appendix

\section{Proof of Theorem \ref{Theo_Frechet}}
\label{Sec_AppCont-and-DiffSepaMaps}

\noindent  In this section, we investigate the differentiability of the approximation maps $( \cdot )^{ \rm tr  }$, $( \cdot )^{ \rm pr }$, $( \cdot )^{ \PCA }$. We adapt  results from  \cite{DetDieKut21} (Theorem~3.4), where differentiability of the approximations on the space of $L^2$-functions is shown. %There, differentiability of the partial trace approximation on the larger space of $L^2$-functions has been validated.
Since the continuous functions form a subspace of $L^2$, endowed with a stronger norm, the Fr\'{e}chet-differentials (for each map) have to coincide on both spaces (if they exist). Therefore, it only remains to show that the ``differential quotients'' converge in the space of continuous functions. Notice that we do not establish positive-semi definitness for any of the approximations, which is well known in the literature (for partial traces see Lemma 2.4 in \cite{Filipiak2018} and for the other approximation types the discussion in \cite{MasSarPan20}).

    \subsection{Differentiability of the partial trace}

    Recall the definition of the partial trace kernels in \eqref{Eq_def_partial_trace_factors} and the partial product approximation in \eqref{Eq_TraceApprox}.

\begin{lem}\label{Lem_TracesContSupNorm}
	Let $K_1 \subset \R^p$, $K_2 \subset \R^q$ be compact sets. Then the (partial) trace operators  
\begin{align*}
    \rm tr  &: \mathcal{ C } ( (K_1 \times K_2 )^2 ) \to \R
    \\
    \rm tr _i &: \mathcal{ C } ( (K_1 \times K_2 )^2 ) \to \mathcal{ C } ( K_i^2 )
\end{align*}
    defined in Section~\ref{Subsubsect_PartTrace}
    are continuous.
\end{lem}
\begin{proof}
	The proof follows by elementary calculations. For $A \in \mathcal{ C } ( ( K_1 \times K_2 )^2)$ we can upper bound the integral
\begin{equation*}
		| {\rm tr}  [ A ]  |
	=
		\left| 
		\int_{  K_1 \times K_2  } A( u, w, u, w )  du \, dw   
		\right|
	\le
		| K_1 |\cdot | K_2 | \cdot \| A \|. %_{ \infty }.
\end{equation*} 
    Here,% with some abuse of notatiox
    $|K_i|$ denotes the Lebesgue measure of $K_i$ in the appropriate dimension.
	 Analogous arguments can be used for the continuity of 
	 $\rm tr _1, \rm tr _2$ using their  definitions in Section~\ref{Subsubsect_PartTrace}.

\end{proof}

\begin{theo}\label{Theo_FrechDiffTr}
	Let $K_1 \subset \R^p$, $K_2 \subset \R^q$ be compact sets.
	Then, the map 
$
    \mathbf{ F }^{ \rm tr }%_{ i }
$ 
    defined in  \eqref{Eq_TraceApprox} is Fr\'{e}chet differentiable in any non-zero covariance kernel $C$. 
\end{theo}
\begin{proof}
	
	The Fr\'{e}chet differentiability of 
$
    \mathbf{ F }_{ i }^{ \rm tr } 
= 
    \mathbf{ H }^{ \rm tr }\circ \mathbf{ G }^{ \rm tr  }, 
$  
    follows by the chain rule from the differentiability of the two maps
$$
	\mathbf{ G }^{  \rm tr  }:  
		\mathcal{ C } ( ( K_1 \times K_2 )^2  ) 
	\to 
		\mathcal{ C } ( K_1^2    ) \times \mathcal{ C } (  K_2^2   )  
\quad
	\text{ and } 
\quad
	\mathbf{ H }^{ \rm tr } : 
		\mathcal{ C } ( K_1^2    ) \times \mathcal{ C } (  K_2^2   ) 
	 \to 
		 \mathcal{ C } ( ( K_1 \times K_2 )^2 )
$$
    defined (point-wise) as
\begin{align*}
&\left(
    \begin{array}{c} 
        \mathbf{ G }_1^{ \rm tr }[ T](s,s') 
    \\
        \mathbf{ G }_2^{ \rm tr }[ T](t,t') 
    \end{array}
	\right) 
=   
    \left(
        \begin{array}{c} 
		     \dfrac{ \int_{ K_2 } T( s, w, s', w )  dw   
    		 }{ 
    		 \int_{ K_1 \times K_2 } T( u, w, u, w ) du \, dw 
	    	 }
		\\
		    \int_{ K_1 } T( u, t, u, t' ) du 
		 \end{array}
	\right) \\[1ex]
& 
        \mathbf{ H }^{ \rm tr } 
        \left[ 
          \left(
            \begin{array}{c} 
	    	    F 
            	\\
	    	     G
	    	\end{array}
		\right) 
	\right](s,t,s',t') 
	= 
	    F( s, s') G( t, t').
	    %T_1 \TePr T_2.
\end{align*}
Using (the proof of) Theorem~3.4 in \cite{DetDieKut21}, we claim that the Fr\'{e}chet differential of $ \mathbf{ G }^{ \rm tr }$ in $C$, i.e., $D_{ C } \mathbf{ G }^{ \rm tr } [T]( s, t, s', t' )$, 
can be point-wise expressed as 
$$
		\left(\begin{array}{c} 
			
			\dfrac{ \int_{ K_2 } T( s, w, s', w ) dw    
		 }{ 
		 \int_{ K_1 \times K_2 } C( u, w, u, w ) du  dw 
		}
			- 
			\dfrac{ \int_{ K_1 \times K_2 } T( u, w, u, w ) du  dw
			\,
			\int_{ K_2 } C( s, w, s', w ) dw    
		 }{ 
		 \left( 
		    \int_{ K_1 \times K_2 } C( u, w, u, w ) du  dw 
		  \right)^2
		 }	
		\\ 
			 \int_{ K_1 } T( u, t, u, t' )  du 
		  \end{array}  
		\right) 
$$
Boundedness of the map $D_{ C } \mathbf{ G }^{ \rm tr } $ follows directly from the continuity of $\rm tr , \rm tr _1, \rm tr _2$ (see Lemma~\ref{Lem_TracesContSupNorm}). Now, we verify that 
$
D_{ C } \mathbf{ G }^{ \rm tr } 
$ 
is indeed the differential. Since the second component is linear and continuous in the sup-norm it must be its own differential. Hence it is enough to consider the first component.
For this purpose, let $ H \in \mathcal{ C } ( ( K_1 \times K_2 )^2 )$ such that $\| H \| \to 0$. Notice that for $H$ sufficiently small, all (partial) traces in the subsequent objects are well-defined. 
A simple calculation yields the representation
\begin{align*}
&	
    \mathbf{ G }_1^{ \rm tr } [ C + H ]( s, t, s', t' ) 
    -  \mathbf{ G }_1^{ \rm tr } [ C ](s, t, s', t') 
    - D_{ C } \mathbf{ G }_1^{ \rm tr }  [ H ](s, t, s', t')
    \\
=
&
	    \frac{ - \left( \int_{ K_2 } H( s, w, s', w ) dw \right)
	    \, \left( \int_{ K_1 \times K_2 } H( u, w, u, w ) du \, dw \right)
		}{
	 	\left( \int_{ K_1 \times K_2 } ( C + H )( u, w, u, w ) du \, dw  \right)
	 	\,
	 	\left( \int_{ K_1 \times K_2 } C( u, w, u, w ) du \, dw \right)
	 	}
	\\
	+&
		\frac{  
		\left(
		\int_{ K_1 \times K_2 } H( u, w, u, w ) du \, dw
		  \right)^2
		\,
		\int_{ K_2 } C( s, w, s', w ) dw 
		}{
		 \left( \int_{ K_1 \times K_2 } ( C + H )( u, w, u, w ) du \, dw \right)
		\,
		\left(
		\int_{ K_1 \times K_2 } C( u, w, u, w ) du \, dw
		  \right)^2 
	}.
\end{align*}
    Arguing as in the proof of Lemma~\ref{Lem_TracesContSupNorm} above one easily sees that
 both terms on the right are uniformly of order $\mathcal{O} \left( \| H \|^2 \right)$, which shows the desired property of Fr\'{e}chet differentiability.
Next we turn to the differentiability of the map $\mathbf{ H }^{ \rm tr }$, where again (the proof of) Theorem~3.4 in \cite{DetDieKut21} suggests the following candidate for  a differential  
$$
		D_{ (  L_1  , L_2  ) } \mathbf{ H }^{ \rm tr }  \left[
		\left(\begin{array}{c} 
			F 
		\\ 
			 G 
		  \end{array}  \right) \right] (s, t, s', t' )= 
	    L_1( s, s' ) G( t, t' ) + F( s, s' ) L_2( t, t' )
	%	  L_1 \TePr T_2 
	%	  +
	%	  T_1 \TePr L_2     
$$
with 
$( L_1  , L_2 )^t := \mathbf{ G }^{ \rm tr } [ C ] $. Proving that  $D_{ (  L_1  , L_2  ) } \mathbf{ H }^{ \rm tr } $ is indeed the Fr\'{e}chet differential follows by similar, but simpler calculations as in  the proof of the differentiability  of $\mathbf{ H }^{ \rm pr }$ in Theorem~\ref{Theo_FrechDiffPr} (below) and is therefore omitted.

\end{proof}

    \subsection{Differentiability of the partial product} 
    \label{Subsec_AppCont-and-Diff-PartProd}

    Recall the definition of the partial product kernels in \eqref{Eq_def_partial_prod_factors} and the partial product approximation in \eqref{Eq_ProdApprox}.

\begin{lem}
	For any $\psi \in \mathcal{ C } ( K_2^2 )$, 
	the  partial product operators as defined in Section~\ref{Subsubsec_PartProd} are continuous.
\end{lem}
\begin{proof}
	The proof is trivial for the linear operator $\mathbf{ F }^{ \rm pr }_1$, which is evidently bounded. In order to show continuity of $\mathbf{ F }^{ \rm pr }_2$, let $A, H \in \mathcal{ C } ( K_2^2 )$ 
	with $\| H  \| \to 0$. Then,
	by definition of its kernel we have
\begin{align} \label{Eq_decomposition_par_prod}
    & 
        \int_{ K_1^2 }  H ( u, t , u', t' ) ( A + H )_1^{ \rm pr }( u, u' ) du \, du' 
        -  
        \int_{ K_1^2 }  A( u, t, u', t' ) 
            \left( ( A + H )_1^{ \rm pr }( u, u' )- A_1^{ \rm pr }( u, u' ) \right) 
        du \, du' 
    \\
    =& 
        \int_{ K_1^2 } \left( 
            H( u, t , u', t' ) H_1^{ \rm pr }( s, s' ) ( u, u' ) 
            +
            H( u, t , u', t' ) A_1^{ \rm pr }( u, u' ) 
            +
            A( u, t , u', t' ) H_1^{ \rm pr }( u, u' ) 
        \right) du \, du'
    \nonumber
\end{align}
    where we used linearity of $\mathbf{ F }^{ \rm pr }_1$ in the second equality. 
    The first integral is of order $\mathcal{O}(\| H \|^2)$ where we have used the (Lipschitz) continuity of the linear map $\mathbf{ F }^{ \rm pr }_1$. The second and third integral are evidently (bounded) linear maps in $H$ corresponding to the derivative.

\end{proof}

\begin{theo}\label{Theo_FrechDiffPr}
	Let $K_1 \subset \R^p$, $K_2 \subset \R^q$ be compact sets and $\psi \in \mathcal{ C } ( K_2^2 )$ be chosen such that $C_1^{ \rm pr }$ (defined in \eqref{Eq_def_partial_prod_factors}) is not identically $0$.
	Then, the map $\mathbf{ F }^{ \rm pr }$ defined in \eqref{Eq_ProdApprox} is Fr\'{e}chet differentiable in $C$.
\end{theo}
\begin{proof}

As in the proof of Theorem~\ref{Theo_FrechDiffTr} above, we can again decompose $\mathbf{ F }^{ \rm pr }$ into two simpler ones
$$
	\mathbf{ G }^{ \rm pr } : 
		\mathcal{ C } ( (K_1 \times K_2 )^2  ) 
	\to 
		\mathcal{ C } ( K_1^2    ) \times \mathcal{ C } (  K_2^2   ) 
 \quad
	\text{ and } 
\quad
	\mathbf{ H }^{ \rm pr } : 
		\mathcal{ C } ( K_1^2    ) \times \mathcal{ C } (  K_2^2   ) 
	 \to 
		 \mathcal{ C } ( (K_1 \times K_2 )^2   )
$$
   point-wise defined as 
\begin{align*}
&
        \mathbf{ G }^{ \rm pr }[ T ]( s, t, s', t' ) 
    = 
        \left(
            \begin{array}{c} 
		        \int_{ K_2^2 } T(s, w, s', w' ) \psi ( w, w' )  dw \, d w' 
        	\\ 
		        \int_{ K_1^2 }  T( u, t, u', t' )
		        \left( 
		            \int_{ K_2^2 } T(u, w, u', w' ) \psi ( w, w' )  dw \, d w' 
		        \right) 
		        d u \, d u' 
		    \end{array}
	   \right)
    \\
&
	\mathbf{ H }^{ \rm pr } \left[ \left(\begin{array}{c} 
		F 
	\\[1ex]
		 G 
		  \end{array}
		  \right) \right]( s, t, s', t' )
	= 
		\dfrac{ F( s, s' )  G ( t, t' ) 
		}{
		\int_{ K_1^2 } ( F ( u,  u' ) )^2  d u \, d u' 
		}.
\end{align*}
 
The derivative in $C$ of the first component of $\mathbf{ G }^{ \rm pr }$ is by linearity and boundedness (w.r.t.\ the sup norm) the map itself. The second component is differentiable, which follows by the decomposition \eqref{Eq_decomposition_par_prod}, where the last two integrals are linear maps of $H$ (the derivative) and the third one is of order $\mathcal{O}(\| H \|^2)$.  
Indeed defining the map (which in fact only depends actually on the variables $(t, t')$)
\begin{align*}		
        D_{ C } \mathbf{ G }^{ \rm pr }_2  [ H ]( s, t, s', t' )  
    =&	
	    \int_{ K_1^2 }  C( u, t, u', t' )
		 \left( 
		    \int_{ K_2^2 } H( u, w, u', w' ) \Delta ( w, w' ) d w \, d w' 
		 \right) d u \, d u' 
    \\
    &+
	    \int_{ K_1^2 }  H( u, t, u', t' )
		 \left( 
		    \int_{ K_2^2 } C(u, w, u', w' ) \Delta ( w, w' )  d w \, d w' 
		 \right) d u \, d u'.
	%P_1 ( C , P_2 ( T, \Delta ) ) + P_1 ( T , P_2 ( C, \Delta ) ).
\end{align*}
and subtracting it from $( \mathbf{ G }^{ \rm pr }_2 [ C + H ] - \mathbf{ G }^{ \rm pr }_2 [ C ] ) (s, t, s', t' )$, we infer as in \eqref{Eq_decomposition_par_prod} that this difference is of the order $\mathcal{O}(\| H \|^2)$.
Next we calculate the derivative  of $\mathbf{ H }^{ \rm pr }$ 
for a generic pair $( L_1, L_2 )^t$, $L_i \in \mathcal{ C } ( K_i^2 )$, $i=1, 2$ with $L_1 \neq 0$. It is point-wise given by %$(T_1, T_2 )^t \mapsto$
\begin{align*}
  	    D_{ ( L_1, L_2 )^t } \mathbf{ H }^{ \rm pr }  
  	    \left[ 
  	        \left(\begin{array}{c} 
	        F 
	        \\[1ex]
		    G 
		    \end{array}
		\right) \right]
		( s, t, s', t' ) 
	=& 
	    \frac{ F( s, s' ) L_2( t, t' )  
            +
            L_1( s, s' ) G( t, t' )
        }{
            \int_{ K_1^2 } ( L_1 ( u,  u' ) )^2 d u \, d u' 
        }
    \\[1ex]
    &-
        \frac{ 2 \left( 
                \int_{ K_1^2 } F( u,  u' ) L_1( u,  u' ) d u \, d u'  
            \right)
            L_1( s, s' ) L_2( t, t' )
        }{
            \left( \int_{ K_1^2 } ( L_1 ( u,  u' ) )^2 d u \, d u'  \right)^2
        }
\end{align*}
Boundedness of the derivative follows directly from boundedness of the kernels $F, G$ (as well as the boundedness away from $0$ of the denominator). Now consider $H_i \in \mathcal{ C }(K_i^2)$ with $\max(\| H_1 \|, \| H_2 \|) \to 0$ differentiability of the map $\mathbf{ H }^{ \rm pr }$ follows by decomposing the difference
\begin{align*}  
   	& 
   	    \left(
   	    \mathbf{ H }^{ \rm pr }
   	        \left[ \left(
                \begin{array}{c} 
	            L_1 + H_1
	            \\[1ex]
	      	    L_2 + H_2
		        \end{array}
		    \right) \right] 
		-	
		\mathbf{ H }^{ \rm pr } 
            \left[ \left(
                \begin{array}{c} 
                L_1 
                \\[1ex]
		        L_2 
                \end{array}
		    \right) \right]
		- 
		D_{ ( L_1, L_2 )^t } \mathbf{ H }^{ \rm pr } 
		    \left[ \left(
		        \begin{array}{c} 
                H_1 
	            \\[1ex]
	            H_2 
		        \end{array}
            \right) \right]
        \right)
        ( s, t, s', t' ) 
    \\
=&
        \left( 
	        H_1( s, s' ) L_2( t, t' ) 
	        + 
	        L_1( s, s' ) H_2(t,t') 
	        - 
	        2 L_1( s, s' ) L_2(t,t') 
	        \dfrac{ 
	             \int_{K_1^2} L_1( u, u' ) H_1( u, u' ) du \, du'  
	           }{
	            \int_{K_1^2} \left( L_1( u, u' ) \right)^2 du \, du'
	           } 
	   \right) 
	\\
	& 
        \times \left( 
            \frac{ 
                1 
            }{
                \int_{K_1^2} \left( ( L_1 + H_1 )( u, u' ) \right)^2 du \, du'
            } 
	       - 
            \frac{
                1
            }{
                \int_{K_1^2} \left( L_1( u, u' ) \right)^2 du \, du'
            }
            \right)
\end{align*}
Evidently, the first factor is of order $\mathcal{O}( \max \{\| H_1 \|, \| H_2 \| \} )$. Furthermore, a small calculation reveals the same rate for the second one. 
This completes the proof of the differentiability for the maps 
$\mathbf{ G }^{ \rm pr }$ and $\mathbf{ H }^{ \rm pr }$ and the differentiability of the partial product approximation follows by the chain rule and the identity 
$
        \mathbf{ F }^{ \rm pr } 
    = 
        \mathbf{ H }^{ \rm pr } \circ \mathbf{ G }^{ \rm pr }
$.

\end{proof}

    \subsection{Differentiability of the SPCA} \label{Subsec_AppCont-and-Diff-SepaPCA}

\noindent  In order to make the following derivations easier to read, we define for a compact set $K$ the space of continuous, symmetric kernels $\mathcal{ C }(K^2)^{Sym}$  as
$$
        \mathcal{ C }( K^2 )^{Sym} 
    := 
        \{A \in \mathcal{ C }( K^2 ): A(x,y) = A(y,x) \,\, \forall x,y \in K\}.
$$
Any kernel $A$ represents the corresponding Hilbert--Schmidt (integral) operator, which (according to the spectral theorem for normal operators acting on $L^2[0,1]$) can be decomposed as
\begin{equation} \label{Eq_eigen_decomposition}
        A(x,y) 
    = 
        \sum_{i \in \N} v_i^{A}(x) v_i^{A}(y) \lambda_i^{A},
\end{equation}
where $\{v_i^{A}\}_{i \in \N}$ are the eigenfunctions and $\{\lambda_i^{A}\}_{i \in \N}$ the corresponding eigenvalues. Without loss of generality, we assume in the following that $|\lambda_1^{A}|\ge |\lambda_2^{A}| \ge |\lambda_i^{A}|$ for any $i \ge 3$. Notice that a priori, the identity \eqref{Eq_eigen_decomposition} only holds true in an $L^2$-sense, but it can be shown that it remains true in the space of continuous functions w.r.t.\ the sup-norm by Mercer's Theorem, see, for instance, Theorem~3.a.1 in \cite{KoeBook86}. Moreover, it can be shown that the eigenfunctions allow the choice of a continuous representative (see Lemma~\ref{Lemma_cont}). Therefore, we will subsequently assume that  $v_i^{A}$ is this representative in $\mathcal{ C }( K )$. 
Let us now assume that $A_0 \in \mathcal{ C }( K^2 )^{Sym}$ is a kernel which satisfies 
$|\lambda_1^{ A_{ 0 } }|> |\lambda_2^{ A_{ 0 } }|$. This means that the first eigenvalue of $A_{ 0 }$ is unique, and the first eigenfunction (corresponding to $\lambda_1^{ A_{ 0 } }$) as well. Since eigenfunctions are generally only determined up to sign, we suppose here that some choice of sign for $v_1^{ A_{ 0 } }$ has been fixed. Then, it follows that there exists a sufficiently small $\delta = \delta(A_0) > 0 $, such that for any $ A \in \mathcal{ C }( K^2 )^{Sym}$ with $ \| A - A_0 \| < \delta $ it holds that $| \lambda_1^{ A } | > | \lambda_2^{ A } | $ making the first eigenvalue of $A$ unique. This also (up to a sign) identifies the first eigenfunction $v_1^{ A }$ of $A$ and we may fix that choice of sign, which minimizes the distance $\int_K (v_1^{A_0}(t)-v_1^{A}(t))^2dt $ (see Lemma~\ref{Lemma_eigen_det} for details). Now, the 
``eigen-maps'' 
\begin{equation} \label{Eq_eigen_maps}
        \Lambda: 
            \begin{cases}
                U_\delta(A_0) \to \R,
            \\
                A \mapsto \lambda_1^{A}\end{cases} 
            \qquad 
                V: \begin{cases}
                U_\delta(A_0) \to \mathcal{ C }( K ),
            \\
                A \mapsto v_1^{A}
            \end{cases}
\end{equation}
are well-defined. Since the separable approximation is (essentially) a rank-$1$-approximation of an operator, the key step in proving differentiability of the SPCA-approximation is proving differentiability of the eigenfunction-map and eigenvalue-map. Since the eigen-maps are known to be differentiable in an $L^2$-sense (again see the proof of Theorem~3.4 in \cite{DetDieKut21}) we only have to validate that the differentials are still the limit of the ``differential quotients'' (as in the above proofs). In the case of the eigenfunction-maps this requires a non-standard representation of the differential, to still guarantee that it maps into the space of continuous functions.

\begin{lem}
	Suppose that $A_0 \in \mathcal{ C }( K^2 )^{Sym}$ satisfies $| \lambda_1^{ A_0 } | > | \lambda_2^{ A_0 } |$. Then for $\delta = \delta(A_0) $ sufficiently small, it holds that the eigen-maps (defined in \eqref{Eq_eigen_maps}) are Fr\'{e}chet-differentiable.
\end{lem}

\begin{proof}
Adapting the proof from Theorem~3.4 in \cite{DetDieKut21} is trivial in the case of the eigenvalue-map, where the differential is given by
$$
        D_{ A_0 } \Lambda[ T ] 
    = 
        \int_{ K^2 } T(x,y) v_1^{ A_{ 0 } }(x) v_1^{ A_{ 0 } }(y)dx \, dy.
$$
In contrast, establishing the result for the eigenfunction-map is more intricate. Therefore, it warrants a detailed discussion. This time we do not start with the differential, but find it easier to derive it in a step-by-step process. For this purpose, consider $H \in \mathcal{ C }(K^2)^{Sym}$ with $\| H \| \to 0$ (in particular, we may assume that 
$A_0 + H \in U_\delta(A_0))$. We now investigate the decomposition
\begin{align*}
		v_1^{ A_0 + H } - v_1^{ A_0 } 
	&=
		\frac{(A_0 + H )  
		}{  
		\lambda_1^{ A_0 + H } 
		}
		\left[ v_1^{ A_0 + H } \right]
		- 
		\frac{A_0  
		}{ 
		\lambda_1^{ A_0 } 
		}
		\left[ v_1^{ A_0 } \right]
	\\
	&=
		\left( 
		    \frac{ 1 
		    }{  
		    \lambda_1^{ A_0 + H } 
		    } 
		    -  
		    \frac{ 1 
		    }{  
		    \lambda_1^{ A_0 } 
		    }  
		\right) 
		( A_0 + H ) \left[ v_1^{ A_0 + H } \right]
		+
		\frac{ (A_0 + H )   
		}{  
		\lambda_1^{ A_0  } 
		} 
		\left[ v_1^{ A_0 + H }  -  v_1^{ A_0 } \right]
		+
		\frac{ H    
		}{  
		\lambda_1^{ A_0  } 
		} 
		\left[ v_1^{ A_0 } \right]
	\\
	& =:
		T_1 + T_2 + T_3
\end{align*}
	Next, we will find the differential of each term separately.
	We start with the first term $T_1$ and note that 
	$
	T_1 
	= 
	- ( \lambda_1^{ A_0 + H } - \lambda_1^{ A_0 } ) v_1^{ A_0 + H } / \lambda_1^{ A_0 + H }$, 
	since $( A_0 + H ) \left[ v_1^{ A_0 + H } \right] = \lambda_1^{ A_0 + H } v_1^{ A_0 + H }$. 
	We claim that the differential is given by 
	$
	H \mapsto - v_1^{ A_0 } / \lambda_1^{ A_0 } 
	\cdot 
	\int_{K^2} H (x,y) v_1^{ A_0 }(x) v_1^{ A_0 }(y) dx \, dy,
	$
	 which is obviously linear and continuous w.r.t.\ the sup-norm. Adding and subtracting cross-terms yields the following bound 
\begin{align*}
		\left\| 
		T_1 
		+
		\frac{  v_1^{ A_0 } 
		}{ 
		\lambda_1^{ A_{ 0 }  } 
		} 
		\int_{K^2} H(x,y) v_1^{ A_0 }(x)  v_1^{ A_0 }(y) dx \, dy   
		\right\|
	&\le 
		\int_{K^2} H(x,y) v_1^{ A_0 }(x)  v_1^{ A_0 }(y) dx \, dy 
		\frac{ \| v_1^{ A_{ 0 } + H } - v_1^{ A_{ 0 }  } \|}{ \lambda_1^{ A_{ 0 }  } }
	\\
	&
	+
		\frac{  \| v_1^{ A_{ 0 }  } \|	 }{ \lambda_1^{ A_{ 0 }  } } 
		\left|  
	    	\int_{K^2} H( x, y) v_1^{ A_0 }(x)  v_1^{ A_0 }(y) dx \, dy 
	    	- 
	    	(  \lambda_1^{ A_{ 0 } + H } -  \lambda_1^{ A_{ 0 }  } ) 
	    \right|
\end{align*}
	Both terms are of order $\mathcal{O} \left( \| H \|^2 \right)$. For the first one we observe that the integral is obviously of order $\mathcal{O} \left( \| H \|^2 \right)$ and the bound of $\mathcal{O} \left( \| H \|^2 \right)$ for the difference of eigenfunctions follows by Lemma~\ref{Lem_bounds_op} part $ii)$. The second term on the right is of order $\mathcal{O} \left( \| H \|^2 \right)$, due to the differentiability of the eigenvalue-map (see the beginning of this proof), as
	$$
	    D_{ A_{ 0 } } \Lambda [ H ] 
	= 
	    \int_{K^2} H( x, y) v_1^{ A_0 }(x)  v_1^{ A_0 }(y) dx \, dy 
	$$
	(for details we refer to \cite{DetDieKut21}).\\
	We continue by analyzing $T_2$. Notice that by the (second) inequality of Lemma~2 in \cite{KokRei13}, one has
$$
	    \| H [  v_1^{ A_{ 0 } + H }  -  v_1^{ A_{ 0 } } ] / \lambda_1^{ A_{ 0 } } \| 
	\le 
	    \frac{ \| H \| 
	    }{
	     | \lambda_1^{ A_{ 0 } } | 
	    }
	    \left( \int_{ K } \left( v_1^{ A_{ 0 } + H }  -  v_1^{ A_{ 0 } } \right)^2 (x ) dx \right)^{ 1 / 2 } 
    < 
        \const \| H \|^2.
$$ 
    This implies that 
$
	    T_2
	=
	    A_{ 0 } [ v_1^{ A_{ 0 } + H }  -  v_1^{ A_{ 0 } } ]  /  \lambda_1^{ A_{ 0 }  } + \mathcal{O}(\| H \|^2)
	$. 
	We can rewrite the non-negligible term as
\begin{align} \label{Eq_DecompKernelOnDiffEigvect}
        \frac{ A_{ 0 } 
        }{
        \lambda_1^{ A_{ 0 }  }
        }
        \left[ v_1^{ A_{ 0 } + H }  -  v_1^{ A_{ 0 } } \right]  
=&
	    \frac{ A_{ 0 } 
        }{
        \lambda_1^{ A_{ 0 }  }
        } 
	    \left[ \sum_{ i \ge 1 } 
	        \left( \int_{ K } \left(v_1^{ A_{ 0 } + H }(x)  
	                            -  
	                         v_1^{ A_{ 0 } }( x ) \right) 
	        v_i^{ A_{ 0 } }(x) dx 
	        \right)
	        v_i^{ A_{ 0 } }
	        \right]
%\\
\end{align}
    Notice that the RHS is nothing else than the $L^2$-basis expansion of
    $ 
    v_1^{ A_{ 0 } + H }  
    -  
    v_1^{ A_{ 0 } }$ w.r.t.\ the ONB 
    $\{ v_i^{ A_{ 0 } } \}_{ i \in \mathbb{ N } }$. This relation definitely holds w.r.t.\ the $L^2$-norm, since $A_0$ is an integral operator, but it needs not hold w.r.t.\ the sup norm. However, we claim it does. First, note that both side are well-defined continuous functions. Moreover, notice that
\begin{align*}
&
   \sup_x
   \left| \int_K A_0 (x, y)
        \left[
         v_1^{ A_{ 0 } + H } (y) 
        -  
        v_1^{ A_{ 0 } } ( y )
        - 
        \sum_{ i \ge 1 } 
	        \left( 
	        \int_{ K } 
	            \left( 
	            v_1^{ A_{ 0 } + H }( u )  
	            -  
	           v_1^{ A_{ 0 } }( u ) 
	           \right) 
	        v_i^{ A_{ 0 } }(u) du 
	        \right)
	        v_i^{ A_{ 0 } } ( y )
	    \right]
	   dy
	\right|
\\
\le&
    \| A_0 \| \int_K  
    \left|
     v_1^{ A_{ 0 } + H } (y) 
        -  
        v_1^{ A_{ 0 } } ( y )
        - 
        \sum_{ i \ge 1 } 
	        \left( 
	        \int_{ K } 
	            \left( 
	            v_1^{ A_{ 0 } + H }( u )  
	            -  
	           v_1^{ A_{ 0 } }( u ) 
	           \right) 
	        v_i^{ A_{ 0 } }(u) du 
	        \right)
	        v_i^{ A_{ 0 } } ( y )
	   dy
    \right| \, = 0
\end{align*}
    where $a_0$ is the continuous kernel of $A_0$.
    Using the first identity and the second identity 
    of Lemma~1 in \cite{KokRei13} on  
    $$
    \left( \int_{ K } 
	   \left( v_1^{ A_{ 0 } + H }(x)  
	   -  
	   v_1^{ A_{ 0 } }( x ) \right) 
	   v_1^{ A_{ 0 } }(x) dx 
	   \right)
	v_1^{ A_{ 0 } } 
\quad
\text{  and } 
\quad
    \sum_{ i \ge 2 } 
	        \left( \int_{ K } 
	            \left(v_1^{ A_{ 0 } + H }(x)  
	             -  
	            v_1^{ A_{ 0 } }( x ) \right) 
	        v_i^{ A_{ 0 } }(x) dx 
	        \right)
	        v_i^{ A_{ 0 } }
$$
	 respectively,
	 one sees that the RHS of 
    \eqref{Eq_DecompKernelOnDiffEigvect} equals 
\begin{align*}
    - 
		\frac{ A_{ 0 } 
        }{
        2 \lambda_1^{ A_{ 0 }  }
        }
        \left[ v_1^{ A_{ 0 } } \right]
        \int_{ K } ( v_1^{ A_{ 0 } + H }(x)  -  v_1^{ A_{ 0 } }(x) )^2 dx 
    +
        \frac{ A_{ 0 } 
        }{
        \lambda_1^{ A_{ 0 }  }
        } 
        \left[ 
			\sum_{ i > 1 } 
			    \frac{ v_i^{ A_{ 0 } } 
			    }{  
			    \lambda_1^{ A_{ 0 } + H } - \lambda_i^{ A_{ 0 } }
			    } 
				\int_{ K^2 } H (x,y) v_1^{ A_{ 0 } + H } (x)  v_i^{ A_{ 0 }  }(y) dx \, dy
		\right]
\end{align*}	 
    where equality holds w.r.t.\ sup-norm by the same argument as above.
    From the proof of Lemma~2 in \cite{KokRei13}, it follows easily that %directly that 

    $
    \int_{K} (v_1^{ A_{ 0 } + H }(x)  -  v_1^{ A_{ 0 } }(x) )^2 dx 
    = 
    \mathcal{O}(\| H \|^2)
    $ so that the first term is negligible. Using a similar reasoning as before (together with the bounds of Lemma~\ref{Lem_bounds_op}), we can then show that
\begin{align*}
& 
        \frac{ A_{ 0 } 
        }{
        \lambda_1^{ A_{ 0 }  }
        } 
        \left[ 
			\sum_{ i > 1 } 
			\frac{ v_i^{ A_{ 0 } } 
			}{
			\lambda_1^{ A_{ 0 } + H } -  \lambda_i^{ A_{ 0 }  }
			} 
			\int_{ K^2 } H (x,y) v_1^{ A_{ 0 } + H }(x)  v_i^{ A_{ 0 }  }(y) dx \, dy
		\right] 
\\
=& 
        \frac{ A_{ 0 } 
        }{
        \lambda_1^{ A_{ 0 }  }
        } 
        \left[ 
			\sum_{ i >1 } 
			\frac{ v_i^{ A_{ 0 } } 
			}{  
			\lambda_1^{ A_{ 0 } } -  \lambda_i^{ A_{ 0 }  } 
			} 
			\int_{ K^2 } H (x,y) v_1^{ A_{ 0 } }(x)  v_i^{ A_{ 0 } }(y) dx \, dy
		\right]
		 + \mathcal{O}(\| H \|^2),
\end{align*}
where we omit the precise calculations to avoid redundancy.
The remaining series in the second line is linear in $H$. A simple calculation shows that it is also bounded. 
Finally, we notice that the term $T_3$ already is a linear, bounded map in $H$. As a consequence, we can rewrite the decomposition
\begin{align*}
& 
        v_1^{ A_{ 0 } + H } - v_1^{ A_{ 0 } } = T_1 + T_2 + T_3 
    =  
        - v_1^{ A_{ 0 } } / \lambda_1^{ A_{ 0 } }  
        \int_{K^2} H (x,y) v_1^{ A_{ 0 } }(x) v_1^{ A_{ 0 } }(y) dx \, dy 
    \\
& 
    +
        \frac{ A_{ 0 } 
        }{
        \lambda_1^{ A_{ 0 }  }
        } 
        \left[ 
			\sum_{ i > 1 } 
			\frac{ v_i^{ A_{ 0 } } 
			}{
			\lambda_1^{ A_{ 0 } } -  \lambda_i^{ A_{ 0 }  }
			} 
			\int_{K^2} H (x,y) v_1^{ A_{ 0 } }(x)  v_i^{ A_{ 0 } }(y) dx \, dy
		\right] 
	+ 
	    \frac{ H    
		}{  
		\lambda_1^{ A_0  } 
		} 
		\left[ v_1^{ A_0 } \right] 
	+ 
	    \mathcal{O}(\| H \|^2).
\end{align*}
	The non-vanishing part is the Fr\'{e}chet-differential. Notice that in $L^2$ this differential can be further simplified to the more common expression  
	$
	\sum_{ i >1 } 
	\frac{ v_i^{ A_{ 0 } } 
	}{ 
	\lambda_1^{ A_{ 0 } } -  \lambda_i^{ A_{ 0 }  }
	} 
	\int_{K^2} H(x,y) v_1^{ A_{ 0 } }(x)  v_i^{ A_{ 0 } }(y) dx \, dy
	$ 
    (where it has been used that 
    $
    A_{ 0 } [ v_i^{ A_{ 0 } } ] 
    = 
    \lambda_i^{ A_{ 0 } } v_i^{ A_{ 0 } } 
    $ 
    in an $L^2$-sense), 
    for instances, see \cite{KokRei13}. However, this expression is not necessarily an element of the continuous functions anymore, since even the continuity of the eigenfunctions may not hold for all $i \in \mathbb{ N }$. These considerations conclude our proof of the differentiability of the eigen-maps.
\end{proof}

\begin{theo}%\label{Theo_FrechDiffSepaPCA}
    Let $K_1 \subset \R^p$, $K_2 \subset \R^q$ be compact sets and $C \in \mathcal{ C }((K_1 \times K_2)^2)$ be a separable covariance operator. Then the map $\mathbf{ F }^{ \PCA }$ as defined in \eqref{Eq_SepPCA-Approx} 
    is Fr\'{e}chet differentiable in $C$ w.r.t.\ the sup-norm.
\end{theo}

\begin{proof} Having established differentiability of the first eigenvalue and eigenfunction of a symmetric operator, the proof now follows step by step as that of Theorem~3.4 of \cite{DetDieKut21}.

\end{proof}

\section{Proof of 
 Theorem~\ref{Theo_MultiplierBootstrap} 
}

\noindent We show weak convergence in the space of continuous functions, relying on the theory of weak convergence for spaces of bounded functions as described in \cite{vdVWelBook96}. For this purpose two conditions need to be verified, tightness and convergence of the marginals, see their Theorem~1.5.4. First, we demonstrate weak convergence of the marginals of the covariance operator, using the Cram\'{e}r--Wold device (see Theorem~29.4 in \cite{BillBook12}). We use classical blocking technique (to handle the bootstrapped part) together with moment inequalities from \cite{Yoshihara1978} for strongly mixing 
random variables. Second, we prove tightness, by establishing asymptotic equicontinuity (see, for instance, Theorem~1.5.7 from \cite{vdVWelBook96}).
The latter is done in turn by controlling entropy bounds. We define the packing numbers and entropy of sets.

 For the sake of brevity, we introduce the following point-wise notations for the objects of interest. So for
 $( s, t ), ( s', t' )  \in K_1 \times K_2$  we set
\begin{equation*}
        \tilde{ X }_i ( s, t )
    := 
        X_i ( s, t ) - \mathbb{ E } X_i ( s, t ) \, ;
    \quad
        \tilde C_N ( s, t, s', t' ) 
    := 
        \frac{ 1 }{ N } 
        \sum_{ i = 1 }^{ N } 
        \tilde{ X }_i ( s, t )  \tilde{ X }_i ( s', t' ),
\end{equation*}
    and 
\begin{align*}
        \tilde B_N^{ ( k ) } ( s, t, s', t' ) 
&
    := 
	    \frac{ 1 }{  N} 
	    \sum_{ i =1 }^{ N} (  
	        \tilde{ X }_i ( s, t ) \tilde{ X }_i ( s', t' )
			- 
			C( s, t, s', t' ) ) \, w_{i,N}.
\end{align*}
The second equality follows by change of summation, where we set $Z^{ ( k ) }_j = 0 $ (deterministic) for any $j \le 0$. In the last  equality, we have defined the (random) weight $w_{i,N}$ in the obvious way. \\
In order to enhance the clarity of our proof we make two simplifications: First, instead of proving convergence of the vector 
$
       \sqrt{ N }  
       (
         \hat{ C }_N  - C  , { B }_N^{ ( 1 ) }, \cdots, 
            { B }_N^{ ( \bl ) } 
       )  
$
we confine ourselves to convergence of their centered versions, i.e., of 
$
       \sqrt{ N }  
       (
         \tilde{ C }_N-C  , \tilde{ B }_N^{ ( 1 ) }, \cdots, 
            \tilde{ B }_N^{ ( \bl ) } 
       )  
$.
Proving that the difference of these vectors is of order $o_{ \mathbb{ P } }(1)$ follows by similar, but simpler techniques as used in the below proof and is therefore omitted. (It is a consequence of the fact that $\tilde{ X }_i$, $i \ge 1$ also satisfies a CLT). Second, w.l.o.g.\ we set $\bl = 1$, since adjusting for $\bl > 1 $ is straightforward and a notational burden.\\
	\fbox{ Step 1: } We begin, proving weak convergence of the finite dimensional distributions, 
   by means of the Cram\'{e}r--Wold device (see, for instance, Theorem~29.4 in \cite{BillBook12}). For this purpose, let  $p \in \N$ be fixed but arbitrary and consider arbitrary tuples $( s_m, t_m, s_m' , t_m') \in (K_1 \times K_2 )^2$ and numbers $a_m, a_m' \in \R$, for $1 \le j \le p$. Recall that to apply the Cram\'{e}r--Wold device, we have to establish weak convergence of the real-valued, random variables
$$
        CW_N 
    :=
        \sqrt{ N } \sum_{m=1}^p 
            \left\{ a_m  
                \left[ \tilde{ C }_N( s_m, t_m, s_m' , t_m') - C( s_m, t_m, s_m' , t_m') \right] 
                + 
                a_m' \tilde B_N^{ ( 1 ) } ( s_m, t_m, s_m' , t_m') 
            \right\}
$$
    to 
$$
    \sum_{m=1}^p 
            \left( a_m  
                G( s_m, t_m, s_m' , t_m') 
                + 
                a_m' G^{ ( 1 ) } ( s_m, t_m, s_m' , t_m') 
            \right),
$$
where $G, G^{ ( 1 ) }$ are two independent identically distributed Gaussian processes. 
We proceed by a blocking technique. It follows, by definition, that we can rewrite $CW_N$ as
\begin{align*}
& 
        CW_N 
    =  
        \frac{ 1 
        }{ 
        \sqrt{N}
        } 
        \sum_{ i =1 }^{ N} 
        \sum_{m=1}^p  
            \left( \tilde{ X }_i ( s_m, t_m ) \tilde{ X }_i ( s_m', t_m' ) 
		    -
		    C( s_m, t_m , s_m', t_m' )
		    \right) 
		    \, ( a_m + a_m' w_{ i, N } ) 
\\
	=&
		\sqrt{ \frac{ b_N }{ N } } 
		\sum_{ j = 1  }^{ \lfloor  N / b_N \rfloor } \sum_{ i = ( j - 1 ) b_N + 1 }^{ j b_N }  \sum_{m=1}^p \left( \tilde{ X }_i ( s_m, t_m ) \tilde{ X }_i ( s_m', t_m' ) 
		-
		C( s_m, t_m , s_m', t_m' )\right)  \, 
		\frac{ ( a_m + a_m' w_{ i, N } ) 
		}{
		\sqrt{b_N} 
		} 
		+ 
		\operatorname{ Rem }_1.
\end{align*}
Here $(b_N)_{ N \in \mathbb{ N } }$ is a sequence of natural numbers, such that $l_N / b_N \to 0$ and $b_N / \sqrt{ N } \to 0$ and $\operatorname{ Rem }_1$ a remainder capturing all ``overhanging terms'' with indices between $b_N \lfloor  N/b_N \rfloor$ and $N$. Using the triangle inequality (and counting terms), it is straightforward to see that 
$
    \mathbb{ E } | \operatorname{ Rem }_1|
=
    \mathcal{O}(b_N/\sqrt{N})=o(1)
$ and hence
$
\operatorname{ Rem }_1
=
o_\mathbb{P}(1)$. 
For ease of reference, we now define the random variables
\begin{align} \label{Eq_def_Y}
        Y_j & 
    :=  
        \sum_{ i =( j -1 ) b_N + 1 }^{ j b_N} 
        \sum_{m=1}^p 
            \left( \tilde{ X }_i ( s_m, t_m ) 
                \tilde{ X }_i ( s_m', t_m' ) 
		    -
		    C( s_m, t_m , s_m', t_m' )
		    \right)  
		    \, 
		    \frac{ ( a_m + a_m' w_{ i, N }) }{ 
		    \sqrt{ b_N } 
		    }
\\
	&= 
	    \sum_{ i = ( j -1 ) b_N + 1 }^{ j b_N}   
	    \, \sum_{m=1}^p \breve Y_{i,m} 
	    \frac{ ( a_m + a_m' w_{ i, N } ) 
	    }{ 
	    \sqrt{ b_N } 
	    }, \nonumber
\end{align}
where $\breve Y_{i,m}$ is defined in the obvious way.  This allows us to write $CW_N = \sqrt{\frac{b_N}{N}} \sum_j Y_j +o_{\mathbb{P}}(1)$. In the following, we demonstrate weak convergence of the (non-negligible) sum, by virtue of the central limit theorem of Wooldridge--White (Theorem~5.20 in \cite{White2001}).  Therefore, we have to check three conditions: Sufficiently fast decay of the mixing coefficients of $(Y_j)_{j=1,\cdots,\lfloor  N/b_N \rfloor }$, uniform boundedness of fourth moments and existence of the asymptotic variance. \\
\textbf{Mixing:} The variables $(Y_j)_{j=1, \cdots, \lfloor  N/b_N \rfloor }$ form a triangular array of $\alpha$-mixing random variables, with mixing coefficients satisfying (for all $N \ge N_0$ and some sufficiently large $N_0 \in \N$) for $|j-j'|\ge 2$
\begin{equation} \label{Eq_mixing_Y}
        \alpha(Y_j, Y_{j'}) 
    \le 
        \kappa 
        \left( 
        ( | j - j' | - 1 ) \, (b_N - l_N + 2 ) \right)^{-a}.
\end{equation}
Here we have used Assumption~\ref{Assum_MainAssumBanachCLT} $iv)$ together with the definition of the random variables $Y_j, Y_{j'}$. More specifically, we have used that $Y_j$ only depends on the functions 
$
\tilde{ X }_{ ( j - 1 ) b_N + 1}, \cdots,
\tilde{ X }_{j b_N}
$ 
and on 
$
%Z_{ ( j - 1 ) b_N - l_N + 2 }, \cdots ,Z_{ j b_N }
w_{ ( j - 1 ) b_N +1 , N }, \cdots, w_{ j b_N, N }
$. 
In particular, the mixing coefficients for $|j-j'|\ge 2$ converge (uniformly) to $0$ and hence the decay condition in the theorem of Wooldridge--White is trivially satisfied.\\
\textbf{Bounded fourth moments:} We now want to show that  
$
\mathbb{ E } [Y_j^4]
\le 
\kappa
<
\infty$  uniformly in $j$ and begin with the following calculation
\begin{align} \label{Eq_moment_bound}
        \mathbb{ E } [Y_j^4] 
    \le& 
        \kappa 
        \sum_{ m = 1 }^p  
        \mathbb{ E } 
        \left[ 
            \left| 
            \sum_{ i = ( j - 1 ) b_N + 1 }^{ j b_N}   
            \, \breve Y_{ i, m } 
            \frac{ ( a_m + a_m' w_{ i, N } ) 
            }{
            \sqrt{ b_N } } \right|^4 \right]
\\
    = 
&
    \kappa 
        \sum_{ m = 1 }^p 
        \mathbb{ E } 
        \left[ 
            \mathbb{ E } 
            \left[ 
                \left|\sum_{ i = ( j - 1 ) b_N+1 }^{ j b_N} \breve Y_{i,m}  
                \, 
                \frac{(a_m+a_m' w_{ i,N } ) 
                }{ 
                \sqrt{ b_N } 
                } 
                \right|^4
                %\bigg| Z_1 , \ldots, Z_N 
                \bigg| w_{ 1, N }, \ldots, w_{ N, N }
            \right] 
        \right] \nonumber
\\
    \le 
&
        \kappa 
        \sum_{m=1}^p  
        \mathbb{ E } 
        \left[ 
            \left( 
            \sum_{ i = ( j - 1 ) b_N+1 }^{ j b_N  } 
                \left(
                \frac{ ( a_m + a_m' w_{ i, N } ) 
                }{
                \sqrt{ b_N } 
                } 
                \right)^2 
            \right)^4 
        \right] 
    \le 
        \kappa. \nonumber
\end{align}
In the first inequality the constant $\kappa$ only depends on $p$. In the first equality, we condition on the $w_{ i, N }$, $1 \le i \le N$ (using the tower property). 
Hence, we can apply Theorem~3 from \cite{Yoshihara1978}, to get the second inequality. Notice that we assume $\breve Y_{ i, m }$ to have a finite $\gamma/2$-moment, which holds since by $iii)$ of Assumptions~\ref{Assum_MainAssumBanachCLT} $X_i$ have finite $\gamma$-moments. In the above cited theorem, we have made the parameter choices $m=4$ and $\delta= \gamma/2 - 4$ and its summability condition $(ii)$ is met since $a> 2 \gamma / ( \gamma - 8 )$ according to $iv)$ of our Assumptions~\ref{Assum_MainAssumBanachCLT}. The constant $\kappa$ in this inequality only depends on the summability of the mixing coefficients (see the proof of Theorem~3 of \cite{Yoshihara1978} for more details) and moments (of $\breve Y_i$), but not on anything else. For the third inequality, we notice that $(a_m+a_m' w_{i,N})^2$ has bounded moments of any order ($w_{i,N}$ is normal with variance 
$1$) and hence by counting terms the last inequality follows, with $\kappa$ now also depending on $a_m, a_m'$, $1 \le m \le p$.\\

\textbf{Convergence of the variance:} Next, we have to prove the long-run variance exists. In order to achieve this, we decompose
\begin{equation} \label{Eq_CW_variance}
        \mathbb{ E } 
        \left[ 
            \left( 
            \sqrt{ \frac{ b_N }{ N } } 
            \sum_j Y_j 
            \right)^2 
        \right] 
    = 
        \frac{ b_N }{ N } 
        \sum_{j=1}^{ \lfloor N / b_N \rfloor }  
        \mathbb{ E } Y_j^2 
    + 
        \frac{ b_N }{ N } \sum_{ | j - j' | = 1 } | \mathbb{ E } Y_j Y_{ j' } | 
    +  
        2 \sum_{h \ge 2} \sup_{ j } | \mathbb{ E } Y_j Y_{ j + h } |.
\end{equation}
Using the mixing condition \eqref{Eq_mixing_Y} we see that the last sum can be bounded (for sufficiently large $N$) by
$$
        \kappa b_N^{-a (\gamma-4)/\gamma} 
        \sup_{j}( \mathbb{ E } |Y_j|^{\gamma/2})^{4/\gamma} \sum_{h \ge 2}  h^{-a (\gamma-4) / \gamma }  
    =
        o(1).
$$
Here we have used a standard covariance inequality for $\alpha$-mixing (see Lemma~3.11 in \cite{DehMikBook02}) to bound $|\mathbb{ E } Y_j Y_{j+h}|$  and for the small $o$-rate, that the sum on the right converges for $a>\gamma/(\gamma-4)$ (see Assumption~\ref{Assum_MainAssumBanachCLT}~$iv)$). Similarly, we can upper bound the other sum of covariances by
$$
        \frac{ b_N }{ N } \sum_{ | j - j' | = 1 } |\mathbb{ E } Y_j Y_{ j' } | 
    \le 
        \kappa  \sup_{j} | \mathbb{ E } Y_j Y_{ j + 1 } |.
$$
We further bound the covariance $|\mathbb{ E } Y_j Y_{j+1}|$ (and show that it converges to $0$ uniformly in $j$). To this end, let $(s_N)_{N \in \N}$ denote an increasing sequence of natural numbers with $l_N/s_N \to 0$ and $s_N/b_N \to 0$. 
Recalling \eqref{Eq_def_Y} we can now split up
\begin{align*}
Y_{j} = & \sum_{ i =j b_N+1 }^{ (j+1)b_N-s_N} \sum_{m=1}^p \breve Y_{i,m} \frac{(a_m+a_m' w_{i,N})}{\sqrt{b_N}}+ \sum_{ i =(j+1)b_N-s_N+1}^{ (j+1)b_N}\sum_{m=1}^p \breve Y_{i,m} \frac{(a_m+a_m' w_{i,N})}{\sqrt{b_N}} =: Y_{j,1}+Y_{j,2}, \\
Y_{j+1} = & \sum_{ i =(j+1) b_N+1 }^{ (j+2)b_N-s_N}  \sum_{m=1}^p\breve Y_{i,m} \frac{(a_m+a_m' w_{i,N})}{\sqrt{b_N}}+ \sum_{ i =(j+2)b_N-s_N+1}^{ (j+2)b_N} \sum_{m=1}^p \breve Y_{i,m} \frac{(a_m+a_m' w_{i,N})}{\sqrt{b_N}} =: Y_{j+1,1}+Y_{j+1,2}.
\end{align*}
Here the variables $Y_{j,1}, Y_{j,2}, Y_{j+1,1}, Y_{j+1,2}$ are defined in the obvious way. Using the same techniques as in \eqref{Eq_moment_bound}, it follows that $\mathbb{ E } [(Y_{j,1})^4], \mathbb{ E } [(Y_{j+1,1})^4]<\infty$ and (counting terms)  $\mathbb{ E } [(Y_{j,2})^4], \mathbb{ E } [(Y_{j+1,2})^4]=o(1)$.
The Cauchy--Schwarz inequality thus implies $\mathbb{ E } [Y_{j,1}Y_{j+1,2}], \mathbb{ E } [Y_{j,2}Y_{j+1,2}], \mathbb{ E } [Y_{j,2}Y_{j+1,1}]=o(1)$. Moreover, using the covariance inequality for $\alpha$-mixing $\mathbb{ E } [Y_{j,1}Y_{j+1,2}] \le \{ \mathbb{ E } [(Y_{j,1})^4]\mathbb{ E } [(Y_{j+1,1})^4] \}^{1/2} \alpha(Y_{j,1}, Y_{j+1,1})^{1/2}$. Now  $ \alpha(Y_{j,1}, Y_{j+1,1}) \le \kappa (s_N-l_N)^{-a} =o(1)$, where we have used Assumption~\ref{Assum_MainAssumBanachCLT} $iv)$ together with the definition of the random variables $Y_{j,1}, Y_{j+1,1}$. Recall therefore, that  $Y_{j,1}$ only depends on the functions $\tilde X_{ j b_N + 1 }, \ldots,\tilde X_{(j+1)b_N-s_N}$ and on the weights 
$w_{j b_N + 1, N }, \ldots, w_{ ( j + 1 ) b_N - s_N, N }$.
These considerations imply that the second and third term, on the right of \eqref{Eq_CW_variance}, i.e., the mixed terms, are of order $o(1)$ and hence the variance is equal to $\frac{b_N}{N} \sum_{j=1}^{\lfloor N/b_N \rfloor} \mathbb{ E } Y_j^2$. Finally, we have to show that this variance convergences. Therefore, let us consider, for $1 < j \le \lfloor N/b_N \rfloor$,  $\mathbb{ E } Y_j^2 $ 
\begin{align} \label{Eq_CW_variance_single}  
=&  
        \mathbb{ E } 
        \left[ 
            \mathbb{ E } 
            \left[ 
                \left( \sum_{ m = 1 }^p \sum_{ i = j b_N + 1  }^{ ( j + 1 ) b_N } \breve Y_{ i, m }  \, 
                \frac{ ( a_m + a_m' w_{ i, N } ) 
                }{
                \sqrt{ b_N } 
                } 
                \right)^2 
                  \Big| w_{ 1, N } \ldots,  
                  w_{ N, N } 
            \right]
        \right]
\\
    =&
        \sum_{ m, l = 1 }^p 
        \frac{ 1 
        }{
         b_N }
        \sum_{ i, k = j b_N + 1 }^{ ( j + 1 ) b_N }  \mathbb{ E } [ \breve Y_{ i, m } \breve Y_{ i, l } ]  \,  \mathbb{ E } [ ( a_m + a_m' w_{ i, N } )( a_l + a_l' w_{ k, N } ) ]
        \nonumber
\end{align}
where we used independence of $\breve Y_{ i, m }$ and $w_{ i, N }$, $1 \le i \le N$. 
Due to our stationarity Assumption~\ref{Assum_MainAssumBanachCLT}, we have $\mathbb{ E } [\breve Y_{i,m} \breve Y_{k,l}] = \tau_Y^{m,l}(|i-k|)$ for a function $\tau_Y^{m,l}: \N \to \R$. Moreover, by construction the covariance of $w_{i,N}$ and $w_{k,N}$ also only depends on $|i-k|$ and $N$. Hence, we can consistently define
$$
    \tau_{w,N}^{m,l}(|i-k|)
:= 
    \mathbb{ E } [(a_m+a_m' w_{i,N}) (a_l+a_l' w_{k,N})  ] =a_m a_l+a_m'a_l' 
    \mathbb{ E } [ w_{i,N} w_{k,N} ] 
$$
which converges as $N \to \infty$ and $\mathbb{ E } [w_{i,N}w_{k,N}] \to 1$ (this follows by definition of the weights; see the very beginning of this proof). We can hence rewrite 
$
    \eqref{Eq_CW_variance_single}$
    as
$$
    \sum_{m,l=1}^p\sum_{|h|<b_N} \tau_Y^{m,l}(|h|) \tau_{w,N}^{m,l}(|h|) (1-|h|/b_N).
$$
Notice that this object does not depend on $j$ and converges to the long-run variance $ \sum_{m,l=1}^p\sum_{h \in \Z} \tau_Y(|h|)^{m,l}$. This latter convergence can be established directly by the dominated convergence theorem. Indeed, first observe that $|\tau_{w,N}^{m,l}(|h|)|\le \kappa$, for $ (1-|h|/b_N) \le 1$. Secondly, the terms $\tau_Y^{m,l}(|h|)$ are summable for any $m,l$, which again follows by the covariance inequality for $\alpha$-mixing random variables. Let $|i-j|=|h|\ge 1$, then
$$
\tau_Y^{m,l}(|h|)=\mathbb{ E } [\breve Y_{i,m} \breve Y_{k,l}]  \le \kappa \{\mathbb{ E } [|\breve Y_{i,m}|^{\gamma/2}]\}^{2/\gamma} \{\mathbb{ E } [|\breve Y_{k,l}|^{\gamma/2}]\}^{2/\gamma} (|h|+1)^{-a (\gamma-4)/\gamma},
$$
which is summable since $a>\gamma/(\gamma-4)$ and due to uniform boundedness of the moments, see $iii)$ of our Assumptions~\ref{Assum_MainAssumBanachCLT}. It follows by straightforward modifications that also $\mathbb{ E } Y_1^2$ converges to the same variance. As a consequence of the above considerations, the variance in \eqref{Eq_CW_variance} converges and we can apply  the central limit theorem of Wooldridge--White (Theorem~5.20 in \cite{White2001}), which entails convergence of the marginal distributions.\\

\noindent It remains to show that 
$
\sqrt{ N } 
\sum_{ p = 1 }^{ N } 
a_m ( \tilde{ C }_N - C ) (s_m, t_m, s_m', t_m' ) 
$ 
and
$
\sqrt{ N } 
\sum_{ p = 1 }^{ N } 
a_m' \tilde{ B }^{ ( 1 ) }_N (s_m, t_m, s_m', t_m' ) 
$
are asymptotically independent and have the same (asymptotic) variance. The asymptotic independence follows readily from the uncorraletedness of the sequences 
$(\tilde{ X }_i )_{ i \in \mathbb{ Z } } $ 
and 
the weights together with the Gaussian limit. As for the equivalence of their variance, this follows by a quick calculation using similar techniques as in the first part of our proof.\\

    Before proceeding to the proof of tightness by asymptotic equicontinuity we state the definition of packing numbers for the sake of completeness.

\begin{defi} 
    Let $( \mathcal{ X }, d )$ be a metric space and $B( x, r)$ a ball of radius $r > 0$ 
    centered around $x \in  \mathcal{ X }$. Then for $\varepsilon > 0$, we define the $\varepsilon$-packing number $D( \varepsilon, d )$ as
$$
    \sup 
    \left\{ n \in \N \, 
        |  
        \, 
        \bigcup_{i = 1}^n B( x_i, \varepsilon ) \supset  \mathcal{ X } 
        \text{ where  $d( x_i , x_j ) > \varepsilon$, $x_i \in \mathcal{ X }$,  $1 \le i \le n$ } 
    \right\}.
$$
\end{defi}

\noindent Note that, clearly, the packing number becomes bigger for smaller $\varepsilon > 0$ and remains finite for any totally bounded sets.  

  \noindent 
In the subsequent part of our proof, for  $K \subset \mathbb{ R }^p$ we set $\rho_K( x, y )
:= 1 \wedge \max_{ i = 1 }^p | x_i - y_i |$.

\noindent	\fbox{ Step 2: }
    We show that the process $\sqrt{ N } (\tilde{ C }_N - C , \tilde{ B }^{ ( 1 ) }_N )$ is asymptotically uniformly $\tilde{ \rho }$-equicontinuous in probability, where 
    $$ 
    \tilde{ \rho } ( (s, t, u, v), (s', t', u', v') )
    := 
    \max\{ 
    \rho_{ K_1 \times K_2 } ( (s, t), ( s', t') )
    , 
    \rho_{ K_1 \times K_2 } ( (u, v), ( u', v') )
    \} 
$$ 
    is our  metric on $( K_1 \times K_2 )^2$. 
    Moreover, recall that by Theorem~1.5.7 in \cite{vdVWelBook96} asymptotic equicontinuity of the process is equivalent to tightness.
    For $\zeta >0$, define the set of pairs 
$$
        \mathcal{ A }_{ \zeta }
    :=
        \{
        ( (s, t, u, v), (s', t', u', v') ) 
        \in 
        ( K_1 \times K_2 )^4
        \,
        |
        \,
        \tilde{ \rho }( 
            %\mathbf{ x}, \mathbf{ x}' 
            (s, t, u, v), (s', t', u', v')
            ) 
        < 
        \zeta
        \}.
$$ 
    We will now bound, for $\epsilon > 0$,
\begin{align} 
&
    \limsup_{ N \to \infty }
    \mathbb{ P } \left\{
    \sup_{ ( \mathbf{ x } , \mathbf{ x }' )
        \in \mathcal{ A }_{ \zeta }
        }
    \left|
        \sqrt{ N }
        (\tilde{ C }_N - C , \tilde{ B }^{ ( 1 ) }_N ) 
        ( \mathbf{ x } )
        %(s, t, u, v) 
        - 
        \sqrt{ N } (\tilde{ C }_N - C , \tilde{ B }^{ ( 1 ) }_N )
        ( \mathbf{ x }' )
        %(s', t', u', v')
    \right|
> 
    \varepsilon
    \right\}
    \nonumber
\\
\le&
    \limsup_{ N \to \infty } 
    \mathbb{ E } \left(
    \sup_{ ( \mathbf{ x } , \mathbf{ x }' )
        \in \mathcal{ A }_{ \zeta }
        }
    \left|
        \sqrt{ N }
        (\tilde{ C }_N - C , \tilde{ B }^{ ( 1 ) }_N ) 
        ( \mathbf{ x } )
        %(s, t, u, v) 
        - 
        \sqrt{ N } (\tilde{ C }_N - C , \tilde{ B }^{ ( 1 ) }_N )
        ( \mathbf{ x }' )
        %(s', t', u', v')
    \right|
    \right)^J
    / \varepsilon^J . \label{Eq_ProbModContStatAndBootstr}
\end{align}
    Using Theorem~2.2.4. in \cite{vdVWelBook96} it is enough to bound the $J$-th moment of the difference in two locations. More precisely, for a $J$ to be specified below, we upperbound 
\begin{equation} \label{Eq_MomCovOpForEquicont}
    \mathbb{ E } \left| 
        \sqrt{ N } 
        ( \tilde{ C }_N - C ) (s, t, u, v ) 
        - 
        \sqrt{ N } 
        ( \tilde{ C }_N - C ) (s', t', u', v' )
    \right|^J
\end{equation}
    and
\begin{equation} \label{Eq_MomBootstrpForEquicont}
    \mathbb{ E } \left| 
        \sqrt{ N }
        \tilde{ B }^{ ( 1 ) }_N (s, t, u, v ) 
        - 
        \sqrt{ N }
        \tilde{ B }^{ ( 1 ) }_N (s', t', u', v' )
    \right|^J .
\end{equation}
    Applying Theorem~3 in \cite{Yoshihara1978} on 
    the $\alpha$-mixing random variables:
    $$
    \tilde{ X }_1 ( s, t )\tilde{ X }_1 ( u, v ) 
    - 
    \mathbb{ E } \left[ \tilde{ X }_1 ( s, t )  \tilde{ X }_1 ( u, v ) \right]
    -
    \tilde{ X }_1 ( s', t' )\tilde{ X }_1 ( u', v' ) 
    + 
    \mathbb{ E } \left[ \tilde{ X }_1 ( s', t' )  \tilde{ X }_1 ( u', v' ) \right]
    $$ 
    with weights $a_i = 1 / \sqrt{ N }$ 
    and (arbitrary, but small) $\delta >0$ 
    we see \eqref{Eq_MomCovOpForEquicont} is less than
\begin{equation}
\label{Eq_BoundApplicationTheo3Yosh}
        \const
        \mathbb{ E }
        \left(
        \left|
        \tilde{ X }_1 (s, t ) \tilde{ X }_1 (u, v )
        -
        \mathbb{ E } \left[ \tilde{ X }_1 ( s, t )  \tilde{ X }_1 ( u, v ) \right]
    -
        \tilde{ X }_1 (s', t' ) \tilde{ X }_1 (u', v' )
    +
    \mathbb{ E } \left[ \tilde{ X }_1 ( s', t' )  \tilde{ X }_1 ( u', v' ) \right]
        \right|
        \right)^{ J }
\end{equation}
    whenever 
   $
   \sum_{ i \ge 1 } 
   (i + 1 )^{ J /2 + 1 } 
   \alpha( i )^{ \delta / ( J + \delta ) }
   <
   \infty.
   $

    The latter term will be bound using continuity properties of our 
     random variables
      as assumed in Assumption~\ref{Assum_MainAssumBanachCLT}$(ii)$. Indeed, taking the expectation of relation~\ref{Eq_HoelderContData}, note we have that 
$$
    | \tilde{ X }_i ( s, t ) - 
    \tilde{ X }_i ( s', t' ) |
    \leq
        ( M + \mathbb{ E } M ) 
        \rho_{ K_1 \times K_2} 
        ( ( s, t), ( s', t' ) )^{ \beta }
        %\max\{ | s - u |, | t - v | \}^{ \beta }
    $$
    so the centered  random variables  are H\"older continuous with as new random constant 
    $\tilde{ M } := M + \mathbb{ E } M$. 
    A quick calculation then shows that
$$
        \left| 
        \tilde{ X }_1 ( s, t ) \tilde{ X }_1 ( u, v )
        -
        \tilde{ X }_1 ( s', t' ) 
        \tilde{ X }_1 ( u', v' )
        \right|
    \le    
        2 \| \tilde{ X }_1 \| 
        \,
        \tilde{ M }
        \tilde{ \rho }^{ \beta } ( (s, t, u, v), (s', t', u', v') ),
$$
    and similarly as above 
\begin{equation*}
        \left| 
        \mathbb{ E } \left(
        \tilde{ X }_1 (s, t ) \tilde{ X }_1 (u, v )
        \right)
        -
        \mathbb{ E } \left(
        \tilde{ X }_1 (s, t ) \tilde{ X }_1 (u, v )
        \right)
        \right|
\le
        2 \mathbb{ E } | \tilde{ M } |  
        \tilde{ \rho }^{ \beta } 
        ( ( s, t, u, v ), ( s', t', u', v' ) ).  
\end{equation*}

   Using the two above bounds, we see that \eqref{Eq_BoundApplicationTheo3Yosh} 
   can be upperbounded by 
$$
   \const
    \mathbb{ E }\left[ \tilde{ M }^J 
    \left\| X_1 \right\|^J \right] 
    \tilde{ \rho }^{ J \beta }
    ( (s, t, u, v ) , ( s', t', u', v' ) ),
$$
where $\const$ may depend on $\beta, J$ but not on $N$.

   To bound \eqref{Eq_MomBootstrpForEquicont} we first condition on the weights
   $ w_{ i, N }^{ ( 1 ) }$, $ 1 \le i \le N$. 
   Then the argument runs along the same lines as the one for bounding
   \eqref{Eq_MomCovOpForEquicont}, namely a straightforward application of Theorem~3 in \cite{Yoshihara1978}. This gives the bound
\begin{align*}
    &
        \const
        \mathbb{ E } 
        \left\{
        \left(
        \sum_{ i = 1 }^N 
            \frac{ 
            \left( 
            w_{ i, N }^{ ( 1 ) } \right)^2 
            }{
            N
            }
        \right)^J
        \mathbb{ E }
        \left[
            \left(
            ( \tilde{ X }_1 \cdot \tilde{ X }_1 - C )
            (s, t, u, v ) 
            - 
            ( \tilde{ X }_1 \cdot \tilde{ X }_1 - C )
            ( s', t', u', v' )
            \right)^J 
        \big|
        w_{ i, N }^{ ( 1 ) }, 1 \le i \le N
        \right]
        \right\}
\\
    \le&
        \const
        \mathbb{ E } \left(
        \sum_{ i = 1 }^N 
            \frac{ 
            \left( 
            w_{ i, N }^{ ( 1 ) } \right)^2
            }{
            N
            }
        \right)^J
        \mathbb{ E } \left( \tilde{ M }^J 
    \left\| X_1 \right\|^J \right) 
    \tilde{ \rho }^{ J \beta } 
    ( (s, t, u, v ), ( s', t', u', v' ) )
\end{align*}
   where we also used independence of $\tilde{ X }_i$, $w_{ i, N }^{ ( 1 ) }$. 
   Another application of Yoshihara's Theorem~3 on the
   $l_N-$dependent sequence $w_{i, N}$ allows us to bound their $J$-th moment, which remains finite since $l_N /N \to 0$, as $N \to \infty$.
    Recall that
    $
    \tilde{ \rho } 
    = 
    \max\{ \rho_{ K_1 }, \rho_{ K_2 } \}
    $ and all norms are equivalent on finite dimensional spaces. Since in general finite dimensional spaces the following bound holds (see, for instance, Ex.~6 of Section~2.1 in \cite{vdVWelBook96})
$$
        D( \varepsilon, \tilde{ \rho }^{ \beta } ) 
    < 
        \frac{ \const %\operatorname{ diam }( K_1 ) \operatorname{ diam }( K_1 )
        }{ \varepsilon^{ 2  (
            d_1 
            + 
            d_2  
            )
            /   \beta 
            }
        }
$$
    (here $\const$ depends on $d_i := \operatorname{dim}( K_i )$, $i=1, 2$ as well as the diameter of $K_i$). 
    Choosing the parameter $J = \lceil 2  ( d_1 + d_2 ) / \beta \rceil + 1  $ and using Markov's inequality together with Theorem~2.2.4 
    in \cite{vdVWelBook96} the expression~\eqref{Eq_ProbModContStatAndBootstr} for any arbitrary $\nu > 0$, is less than
$
    ( \const / \varepsilon^J ) 
    ( \eta^{  - 2  (d_1 + d_2  )/ ( J \beta ) + 1 } 
    + 
    \zeta \eta^{  - 4  (d_1 + d_2  )^2/ (J \beta^2) } )
$   
    which can be made arbitrarily small picking $\zeta$ small and then $\nu$ small.
    Consequently, our process is $\tilde{ \rho }$-equicontinuous in probability.

\newpage

\section{Additional results}

\noindent Throughout this section, we always assume that the eigenvalues of an operator $A$ satisfy $|\lambda_1^{A}|\ge |\lambda_{2}^{A}| \ge |\lambda_{i}^{A}| $ for any $i \ge 3$.

\begin{lem}\label{Lemma_cont}
    Let $A \in \mathcal{ C }(K^2)^{Sym}$ be a kernel, with eigenvalues $|\lambda_1^{A}|>|\lambda_{2}^{A}|$, then we can find a continuous representative of the first eigenfunctions $v_1^{A} \in \mathcal{ C }(K)$.
\end{lem}

\begin{proof}
    We first notice that in an $L^2$-sense the equality $v_1^{A} = A[v_1^{A}]/\lambda_1^{A}$ holds. Now defining $v_1^{A}$ point-wise by the expression $A[v_1^{A}]/\lambda_1^{A}$, we see that it is already continuous, as
    \begin{align*}
    & v_1^{A}(t+h)-v_1^{A}(t) = (\lambda_1^{A})^{-1}\int_K [A(s,t+h)-A(s,t)]v_1^{A}(s)ds\\
    \le &  (\lambda_1^{A})^{-1} \Big\{\int_K [A(s,t+h)-A(s,t)]^2 ds\Big\}^{1/2} \le \const \sup_{t}|A(s,t+h)-A(s,t)| =o(1).
    \end{align*}
    Here we have used  Cauchy--Schwarz  in the first inequality. The small-$o$ refers to convergence as $|h|\to 0$ and follows because the continuous kernel $A$ is uniformly continuous on the compact set $K^2$. 
\end{proof}

\begin{lem} \label{Lem_bounds_op}
	Suppose that $A_0 \in \mathcal{ C }(K^2)^{Sym}$ with $A_0$ satisfying $|\lambda_1^{A_0}|>|\lambda_2^{A_0}|$. Furthermore, consider for some $\delta>0$ the operator $A \in \mathcal{ C }(K^2)^{Sym}$ with $\| A - A_0 \| \le  \delta$ and a choice of eigenfunction s.t.\ $\int_K v_1^{A_0}(t) v_1^{A}(t)dt \ge 0$. Then it holds with a constant $\const := \const (A_0, \delta)$

\begin{itemize}
\item[ i) ] For $j=1,2$
\begin{equation*} 
		\left| |\lambda_j^A| - |\lambda^{A_0}_j| \right|
	\le 
		\const \| A-A_0 \|.
\end{equation*}

	\item[ ii) ] 
\begin{equation*} 
		\left\{\int_K (v_1^A(t) - v_1^{A_0}(t))^2 dt \right\}^{1/2}
	\le 
		 \const \| A - A_0\|.
\end{equation*}
\end{itemize}
\end{lem}
\noindent The identities follow directly from $i)$ in \cite{HorKokBook12} (Lemmas~2.2-3), applied to the operators $A$ and $A_0$.
Notice that we have here exploited that the sup-norm is stronger than the $L^2$-norm.

\begin{lem}\label{Lemma_eigen_det}
    Let $A_0 \in \mathcal{ C }(K^2)^{Sym}$ be a kernel, with eigenvalues $|\lambda_1^{A_0}|>|\lambda_{2}^{A_0}|$ and eigenfunction $v_1^{A_0}$ (where some choice of sign for the eigenfunction is fixed). Then there exists a $\delta=\delta(A_0)>0$ sufficiently small, s.t.\ for any $A \in \mathcal{ C }(K^2)^{Sym}$ with $\|A-A_0\| \le \delta$ it holds that $|\lambda_1^{A}|>|\lambda_{2}^{A}|$ and some choice of sign exists for $v_1^{A}$ s.t. $\int_K (v_1^{A_0}(t)-v_1^{A}(t))^2dt <\int_K (v_1^{A_0}(t)+v_1^{A}(t))^2dt$.
\end{lem}

\begin{proof}
    The proof is a direct consequence of the preceding Lemma~\ref{Lem_bounds_op}. Notice that
    \begin{align*}
        |\lambda_1^{A}|=(|\lambda_1^{A}|-|\lambda_1^{A_0}|)+(|\lambda_1^{A_0}|-|\lambda_2^{A_0}|)+(|\lambda_2^{A_0}- |\lambda_{2}^{A}|)+|\lambda_{2}^{A}| \ge |\lambda_{2}^{A}|-\kappa \delta,
    \end{align*}
    where $\kappa$ comes from Lemma~\ref{Lem_bounds_op}. For sufficiently small $\delta>0$, the inequality $|\lambda_1^{A}|>|\lambda_{2}^{A}|$ holds. Finally, the inequality $\int_K (v_1^{A_0}(t)-v_1^{A}(t))^2dt <\int_K (v_1^{A_0}(t)+v_1^{A}(t))^2dt$ follows directly from part $ii)$ of Lemma~\ref{Lem_bounds_op}, again for small enough $\delta$.
\end{proof}

	\end{document}